\documentclass{article}
\usepackage[british]{babel}
\usepackage[colorlinks=true,linkcolor=blue,citecolor=blue]{hyperref}

\usepackage{color,array,colortbl}

\usepackage{graphicx}
\usepackage{epsfig}
\usepackage{multirow}
\usepackage{xcolor}
\usepackage{subfigure}
\usepackage{amsmath,amsfonts,amsthm,bm}
\usepackage[a4paper]{geometry}
\usepackage{fullpage}
\usepackage{algorithm,algorithmic}

\newcolumntype{A}{>{\columncolor[gray]{0.8}}c}
\newcolumntype{B}{>{\columncolor[gray]{0.6}}c}

\DeclareMathOperator{\var}{Var}

\DeclareMathOperator{\argmin}{argmin}

\newcommand{\bsu}{\boldsymbol{u}}    
\newcommand{\bsv}{\boldsymbol{v}}    
\newcommand{\bsw}{\boldsymbol{w}}    
\newcommand{\bsz}{\boldsymbol{z}}    
\newcommand{\bsS}{\boldsymbol{S}}    
\newcommand{\bsW}{\boldsymbol{W}}    
\newcommand{\bsU}{\boldsymbol{U}}    
\newcommand{\bsDelta}{\boldsymbol{\Delta}}    

\def\citep#1#2{\cite[{#1}]{#2}}

\theoremstyle{plain}
  \newtheorem{theorem}{Theorem}
  \newtheorem{lemma}{Lemma}
  
\theoremstyle{definition}

\theoremstyle{remark}
  \newtheorem*{remark}{Remark}

\newcommand{\RefEq}[1]{~\textup{(\ref{#1})}}
\newcommand{\RefEqEq}[2]{~\textup{(\ref{#1})} and~\textup{(\ref{#2})}}

\newcommand{\RefSec}[1]{Section~\textup{\ref{#1}}}
\newcommand{\RefSecSec}[2]{Sections~\textup{\ref{#1}} and~\textup{\ref{#2}}}

\newcommand{\RefThm}[1]{Theorem~\textup{\ref{#1}}}

\newcommand{\RefAlg}[1]{Algorithm~\textup{\ref{#1}}}
\newcommand{\RefFig}[1]{Figure~\textup{\ref{#1}}}

\newcommand{\RefTab}[1]{Table~\textup{\ref{#1}}}

\begin{document}

\title{Conditional sampling for barrier option pricing \linebreak under the LT method}
\author{Nico Achtsis \and Ronald Cools \and Dirk Nuyens}
\maketitle


\begin{abstract}
We develop a conditional sampling scheme for pricing knock-out barrier options under the Linear Transformations (LT) algorithm from Imai and Tan (2006), ref.\ \cite{IT2006}.
We compare our new method to an existing conditional Monte Carlo scheme from Glasserman and Staum (2001), ref.\ \cite{GS2001}, and show that a substantial variance reduction is achieved.
We extend the method to allow pricing knock-in barrier options and introduce a root-finding method to obtain a further variance reduction.
The effectiveness of the new method is supported by numerical results.
\end{abstract}


\section{Introduction}
Barrier options are financial instruments whose payoff depend on an underlying asset hitting a specified level or not during some period.  The simplest example would be a regular knock-out call option, where the payoff is set to zero if the option hits a certain barrier level during the lifetime of the option, and results in the payoff of a regular call option otherwise. These options are popular because they cost at most the amount of their vanilla counterparts as there is a higher probability of ending up with nothing. For more background information on barrier options and their applications, see for instance \cite{DK1996}, \cite{DK1997}, \cite{Wilmott2006} and \cite{Wystup2006}.
For sufficiently simple products there exist analytical formulas for the value, see, e.g., all the previous references, or \cite{Hull2005}. When the product is too complex, one must resort to numerical methods to obtain the value. Furthermore, if the barrier is not monitored continuously but discretely, no accurate closed form formulas are available. Several papers exist dealing with discretely monitored barrier options, e.g., \cite{JT2010} and \cite{ZFV1998}.

A straightforward simulation method for pricing discretely monitored options is to sample the assets in each monitoring date. However, when a barrier condition is introduced in the payoff, only a minority of the paths may lead to non-zero values if the probability of satisfying the barrier condition is low. This leads to significant increases in the variance of the estimator. Glasserman and Staum \cite{GS2001} remedy this problem by introducing a sampling scheme where the asset with the barrier condition is sampled conditional on survival in each monitoring date. In \cite{JT2010} a different scheme is used based on the hitting time of the asset as the sampling variable instead of the asset increments. Both of these papers use Monte Carlo to sample the necessary variates.

In this paper we use low-discrepancy points to generate the uniform variates for sampling the asset paths.  
The merit of using low-discrepancy points for valuing financial products is discussed in \cite{PP1996} and \cite{PT1995}.
Since then quasi-Monte Carlo methods have often been applied to applications in finance, e.g., \cite{LL2000,CKN2006,GKSW2008,LEC2009,Lem2009,NW2012}.
Care must be taken however in the construction of the asset paths, see for instance \cite{Papa2002} or \cite{WS2011}, where it is shown how the path construction method can have a large influence on the convergence rate, and also \cite{WT2012}, which discusses the influence on the nature of the discontinuity.
Imai and Tan \cite{IT2006} provide an alternative to the Brownian Bridge or PCA construction, called the Linear Transformation method (or LT for short).
They provide an optimal construction, in some sense, by minimizing the effective dimension of the underlying payoff function.
We reconcile this LT method with the idea of conditional sampling, to gain a ``best of two worlds'' approach for valuing barrier options.
Even though our method can be applied using other covariance matrix decompositions, such as PCA and Cholesky, we will use the LT method because of its proven performance, see \cite{IT2006}.

The outline of our paper is as follows. 
In \RefSecSec{sec:LTmethod}{sec:CondSampl} we present an overview of the LT algorithm from \cite{IT2006} and the conditional sampling scheme from \cite{GS2001}. 
\RefSec{sect:LTCS} contains the main results of our paper: a modified conditional sampling scheme compatible with the LT algorithm. 
In \RefSec{sec:numres} we give several numerical examples to illustrate our method.
Then in \RefSecSec{sec:knock-in}{sec:root-finding} we present two extensions to our sampling scheme. The first extension deals with knock-in options, and the second one is a modification to combine the quasi-Monte Carlo method with an analytical integration using root-finding.
We summarize our results in \RefSec{sec:conclusion}.


\section{The LT method}\label{sec:LTmethod}

We assume a Black--Scholes world in which the risk-neutral dynamics of the assets are given by
\begin{align*}
 dS_i(t)
  &=
  rS_i(t)dt + \sigma_iS_i(t)dW_i(t)
  ,
  \qquad
  i=1,\ldots,n,
\end{align*}
where $S_i(t)$ denotes the price of asset $i$ at time $t$, $r$ is the risk-free interest rate and $\sigma_i$ the volatility of asset $i$.
Also, $\bsW=(W_1(t),\ldots,W_n(t))$ is an $n$-dimensional Brownian motion, with $dW_idW_j=\rho_{ij}dt$.
For more information on this type of market model, see for instance \cite{Fries2007}, \cite{Glass2003}, \cite{Hull2005} or \cite{Wilmott2006}.
When resorting to Monte Carlo techniques for pricing options under this model, asset paths need to be sampled at each time point. 
In what follows, we assume an equally spaced time discretization $\Delta t=T/m$ for simplicity, but we note that all results can be applied without using this assumption.
Under this assumption, we have $t_j=j\Delta t$.
Under Black--Scholes dynamics we simulate the trajectories according to
\begin{align*}
 S_i(t)
  &=
  S_i(0)e^{(r-\sigma_i^2/2)t+\sigma_iW_i(t)}
  .
\end{align*}
Write $\widetilde{\bsW}=(\sigma_1W_1(t_1),\ldots,\sigma_1W_1(t_m),\sigma_2W_2(t_1),\ldots,\sigma_nW_n(t_m))'$, where we use the prime to denote the transpose of a vector.
Then $\widetilde{\bsW}$ is multivariate normally distributed with covariance matrix
\begin{align*}
 \widetilde{\Sigma}
  &=
  \begin{pmatrix} \Sigma_{11} & \Sigma_{21} & \cdots & \Sigma_{n1} \\
                  \Sigma_{12} & \Sigma_{22} & \cdots & \Sigma_{n2} \\
                  \vdots & \vdots & \ddots & \vdots \\
                  \Sigma_{1n} & \Sigma_{2n} & \cdots & \Sigma_{nn}
  \end{pmatrix}
  ,
\end{align*}
where 
\begin{align*}
 \Sigma_{ij}
 &=
 \begin{pmatrix} 
 1 & 1 & \cdots & 1 \\
 1 & 2 & \cdots & 2 \\
 \vdots & \vdots & \ddots & \vdots \\
 1 & 2 & \cdots & m 
 \end{pmatrix}\rho_{ij}\sigma_i\sigma_j\Delta t.
\end{align*}
One approach to simulate $\widetilde{\bsW}$ is to decompose the matrix $\widetilde{\Sigma}$ as
\begin{align*}
 CC'
  &=
  \widetilde{\Sigma}
\end{align*}
where $C$ is the Cholesky factor (see, e.g., \cite{Glass2003}) and calculate
\begin{align}
 \widetilde{\bsw}
  &=
  C\bsz
  \label{eq:cholfac}
\end{align}
where $\bsz$ is an $mn$-dimensional vector of i.i.d.\ standard normally distributed variables.
Assuming that we have a European option payoff represented as
\begin{align*}
 &
  \max\left( f(\widetilde{\bsW}),0\right)
\end{align*}
the standard Monte Carlo method simulates the function $f(\widetilde{\bsW})$ by mapping a uniform variate $\bsu$ in the unit cube to $\bsz$ by applying the inverse cumulative distribution function $\Phi^{-1}$. 
When using quasi-Monte Carlo, the uniform variates are replaced by a low-discrepancy point set.
The matrix $C$ in\RefEq{eq:cholfac} can be multiplied with any orthogonal matrix $Q$ while retaining the correct distribution for $\widetilde{\bsw}$, see, e.g., \cite{Glass2003}.
Under the LT method described in \cite{IT2006}, instead of simulating $f$ directly from $\bsz$, first an orthogonal transformation is applied, i.e., $f$ is simulated from
\begin{align}\label{eq:AisCQ}
 \widetilde{\bsW}
  &=
  A\bsz  
  =
  CQ\bsz
\end{align}
for a carefully chosen orthogonal matrix $Q$. 
We remark that, for ease of notation, we will write $f(\bsz)$,$f(\bsu)$ or $f(\hat{S}_1(t_1), \ldots, \hat{S}_n(t_m))$ to denote the function $f$ from above in terms of normal variates $\bsz$, uniform variates $\bsu$ or just the stock paths $\hat{S}_1(t_1), \ldots, \hat{S}_n(t_m)$.
In what follows we use the notation $Q_{\bullet k}$ to denote the $k$th column of $Q$ and $Q_{k \bullet}$ to denote the $k$th row.
The matrix $Q$ is chosen according to the following optimization problem:
\begin{align*}
 \underset{Q_{\bullet k}\in\mathbb{R}^{mn}}{\text{maximize}} \qquad& \text{variance contribution of $f$ due to $k$th dimension} \\
  \text{subject to} \qquad& \lVert Q_{\bullet k} \rVert=1, \\
		    \qquad& \langle Q^*_{\bullet j},Q_{\bullet k} \rangle=0, \quad j=1,\ldots,k-1,
\end{align*}
where $Q^*_{\bullet j}$ denotes the columns of $Q$ that have already been optimized in the previous iterations.
The algorithm is carried out iteratively for $k=1,2,\ldots,mn$ so that in the $k$th optimization step the objective function ensures that, given columns $Q^*_{\bullet j}$, $j=1,\ldots,k-1$ which have already been optimally determined in the previous iterations, the variance contribution due to the $k$th dimension is maximized while the constraints ensure orthogonality.
Being able to express the variance contribution for each component analytically for general payoff functions $f$ can be quite complicated.
Therefore, Imai and Tan \cite{IT2006} propose to approximate the objective function by linearizing it using a first-order Taylor expansion around a point $\bsz=\hat{\bsz}+\Delta \bsz$,
\begin{align*}
 f(\bsz)
  &\approx
  f(\hat{\bsz}) + \sum_{\ell=1}^{mn} \left.\frac{\partial f}{\partial z_{\ell}}\right|_{\bsz=\hat{\bsz}}\Delta z_{\ell}. 
\end{align*}
Using this expansion, the variance contributed due to the $k$th component is 
\begin{align*}
 \left( \left.\frac{\partial f}{\partial z_k}\right|_{\bsz=\hat{\bsz}}\right)^2.
\end{align*}
The expansion points are chosen as $\hat{\bsz}_k = (1,\ldots,1,0,\ldots,0)$, the vector with $k-1$ leading ones.
The optimization problem becomes
\begin{align*}
 \underset{Q_{\bullet k}\in\mathbb{R}^{mn}}{\text{maximize}} \qquad& \left( \left.\frac{\partial f}{\partial z_k}\right|_{\bsz=\hat{\bsz}_k}\right)^2 \\
  \text{subject to} \qquad& \lVert Q_{\bullet k} \rVert=1, \\
		    \qquad& \langle Q^*_{\bullet j},Q_{\bullet k} \rangle=0, \quad j=1,\ldots,k-1.
\end{align*}
In \cite{IT2006}, only the first $25$ or $50$ columns are computed.
The loss of efficiency is argued to be minimal, particularly when the underlying LT construction is effective at dimension reduction with the first few dimensions already capturing most of the variance. However, as this construction is just part of the startup cost, we construct the entire matrix in our numerical tests.

\bigskip

The original paper, \cite{IT2006}, considers a basket arithmetic average option (Asian basket) to illustrate the computational advantage of the LT method.
The payoff function inside the $\max$-function in this case is 
\begin{align*}
 f(\bsz)
  &=
  \sum_{i=1}^{mn} e^{\mu_i + \sum_{j=1}^{mn} a_{ij}z_j} - K
\end{align*}
where $K$ is the strike price, and
\begin{align*}
 \mu_i
  &=
  \log(w_{i_1,i_2}S_{i_1}(0)) + \left(r-\frac{\sigma_{i_1}^2}{2}\right)t_{i_2}
  ,
\end{align*}
where $w_{i_1,i_2}$ corresponds to the weight given to asset $i_1=\lfloor (i-1)/m\rfloor+1$ at time $i_2 = i-(i_1-1)m$.

From the first-order Taylor expansion we have
\begin{align*}
 f(\bsz)
  &\approx
  f(\hat{\bsz}) + \sum_{\ell=1}^{mn} \left(\sum_{i=1}^{mn} e^{\mu_i + \sum_{j=1}^{mn} a_{ij}\hat{z}_j} \right)\Delta z_{\ell}.
\end{align*}
The optimization problem becomes
\begin{align*}
 \underset{Q_{\bullet k}\in\mathbb{R}^{mn}}{\text{maximize}} \qquad& \left( \sum_{i=1}^{mn} e^{\mu_i + \sum_{j=1}^{k-1} \langle C_{i\bullet},Q^*_{\bullet j}\rangle}\langle C_{i\bullet},Q_{\bullet k}\rangle\right)^2 \\
  \text{subject to} \qquad& \lVert Q_{\bullet k} \rVert=1, \\
		    \qquad& \langle Q^*_{\bullet j},Q_{\bullet k} \rangle=0, \quad j=1,\ldots,k-1.
\end{align*}
For $k=1$, the solution is found as
\begin{align}
 Q^*_{\bullet 1}
 &=
 \pm\frac{\sum_{i=1}^{mn}e^{\mu_i}C_{i\bullet}}{\lVert \sum_{i=1}^{mn}e^{\mu_i}C_{i\bullet} \rVert}
 ,
 \label{eq:ITk1sol}
\end{align}
for subsequent $k$ the solution is
\begin{align*}
 Q^*_{\bullet k} 
  &=
  \pm\frac{ \sum_{i=1}^{mn}  e^{\mu_i + \sum_{j=1}^{k-1} \langle C_{i\bullet},Q^*_{\bullet j}\rangle}\langle C_{i\bullet},Q_{\bullet k}\rangle }{\lVert \sum_{i=1}^{mn}  e^{\mu_i + \sum_{j=1}^{k-1} \langle C_{i\bullet},Q^*_{\bullet j}\rangle}\langle C_{i\bullet},Q_{\bullet k}\rangle \rVert}
  ,
\end{align*}
where the sign can be chosen freely. The proof is in the appendix of \cite{IT2006}. 
Later in \RefSec{sect:LTCS} we revisit the sign choice.


\section{Conditional sampling according to Glasserman \& Staum}\label{sec:CondSampl}

Suppose we are interested in pricing an up-\&-out option on $n$ assets with payoff 
\begin{align*}
 g(S_1(t_1),\ldots,S_n(t_m))
 &=
 \max\left(f(S_1(t_1),\ldots,S_n(t_m)),0\right)\,\mathbb{I}\left\{\max_jS_1(t_j)<B\right\}
 ,
\end{align*}
where $\mathbb{I}$ is the indicator function.
Only the first asset is included inside the indicator function for notational and expositional ease. 
Our method can handle barriers on multiple assets by a straightforward extension; see also the remark after \RefAlg{alg:ltcs}.

The valuation can be done by sampling the assets using Monte Carlo and then calculating the payoff $g(S_1(t_1),\ldots,S_n(t_m))$.
When the probability of survival becomes low (e.g., if $B-S_1(t_0)$ is small) a lot of the generated paths will result in a knock-out.
This can make the variance among all paths quite large relative to the price.
To remedy this problem Glasserman and Staum \cite{GS2001} propose an estimator using conditional sampling.
Given $S_1(t_j)$ for $j=0,\ldots,i-1$, define 
\begin{align}\label{eq:likelihood}
	L_i
	=
	L_i(S_1(t_0), \ldots, S_1(t_{i-1}))
	&=
	\prod_{j=0}^{i-1} \mathbb{P}\left[ S_1(t_{j+1})<B | S_1(t_j) \right]
	.
\end{align}
The variable $L_i$ can be interpreted as the ``likelihood" of the path surviving $i$ steps.
Under the Black--Scholes framework, the probability in the product can be analytically determined as
\begin{align}\label{eq:Gammaj}
	\mathbb{P}\left[ S_1(t_{j+1})<B | S_1(t_j) \right]
	&=
	\underbrace{\Phi\left[\frac{\log\left(\frac{B}{S_1(t_j)}\right)-\left(r-\frac{\sigma_1^2}{2}\right)\Delta t}{\sigma_1\sqrt{\Delta t}}\right]}_{=\Gamma_j(S_1(t_j))}
	.
\end{align}
The authors \cite{GS2001} now sample the variable
\begin{align*}
	\hat{g}_1
	&=
	L_m \max\left(f(S_1(t_1),\ldots,S_n(t_m)),0\right)
	,
\end{align*}
where all $S_1(t_j)<B$, $j=1,\ldots,m$.
Therefore, there are no paths knocked out in the simulation, greatly reducing variance if such an event has a high probability.
An overview of this algorithm is presented in \RefAlg{alg:gscs}.
We recapitulate two theorems from \cite{GS2001}, which we prove using methods we will use again later.
\begin{theorem}
		The estimator based on conditionally sampling the asset path such that all $S_1(t_j) < B$, $j=1,\ldots,m$, by
		\begin{align}
		\hat{g}_1
		&=
		L_m \max\left(f(S_1(t_1),\ldots,S_n(t_m)),0\right)
		\label{eq:GSVhat}
	\end{align}
	 is unbiased. That is, 
	$\mathbb{E}[\hat{g}_1]=\mathbb{E}[g]$.
	\label{Thm:MCCSunbiased}
\end{theorem}
\begin{proof}
Without loss of generality, we only consider the case of one asset and set $r=0$. 
Denote by $\chi$ the function expressing the barrier condition
\begin{align*}
 \chi(\bsU)
 &=
 \begin{cases} 1 & \text{if $\max_{j \in\{1,\ldots,m\}} S(t_j)<B$,}
  \\
  0 &\text{otherwise,}
 \end{cases}
\end{align*}
where we explicitly denote the dependence of $\chi$ on $\bsU$, the uniform variates used to sample the asset paths.
To stress the dependency of $f$ and $\Gamma_j$ on $\bsU$ we will write this explicitly as well in terms of $\bsU$. 

We find
\begin{align*}
 \mathbb{E}\left[\chi(\bsU) \max(f(\bsU),0)\right]
 &= \int_{[0,1]^{m}} \chi(\bsu) \max(f(u_1,\ldots,u_{m}),0) \,\, d\bsu \\
 &= \int_0^{\Gamma_1}\cdots\int_0^{\Gamma_m(u_1,\ldots,u_{m-1})} \max(f(u_1,\ldots,u_{m}),0) \,\, du_m\cdots du_1
 .
\end{align*}
By using the change of variables $u_1/\Gamma_1=\hat{u}_1$ to $u_m/\Gamma_m(\hat{u}_1,\ldots,\hat{u}_{m-1})=\hat{u}_m$, with $\Gamma_j$ as defined in\RefEq{eq:Gammaj}, we obtain
\begin{align*}
 \mathbb{E}\left[\chi(\bsU) \max(f(\bsU),0)\right]
 &= \int_{[0,1]^{m}} \Gamma_1\prod_{j=2}^m\Gamma_j(\Gamma_1,\ldots,\Gamma_{j-1}(\hat{u}_1,\ldots,\hat{u}_{j-1}))\\ &\qquad\qquad\times\max(f(\Gamma_1\hat{u}_1,\ldots,\Gamma_m(\hat{u}_1,\ldots,\hat{u}_{m-1})u_{m}),0) \,\, d\hat{u}_1\cdots d\hat{u}_{m}\\
 &= \mathbb{E}\left[L_m\max(f(\hat{\bsU}),0)\right]
 ,
\end{align*}
which is what we needed to prove.
\end{proof}

Furthermore, it can be proven that the variance of the new estimator is at most that of the original one.

\begin{theorem}
	The estimator using\RefEq{eq:GSVhat} has reduced variance. That is, 
	$\var[ \hat{g}_1] \leq \var[g ] .$
	The inequality is strict if $\mathbb{P}\left[\max_jS_1(t_j)\geq B\right]>0$ and $\mathbb{E}[g]>0$, i.e., if there is any chance of knock-out and positive payoff.
\end{theorem}
\begin{proof}
Because of \RefThm{Thm:MCCSunbiased} we only need to prove that $\mathbb{E}[\hat{g}_1^2]\leq\mathbb{E}[g^2]$.
We assume again that $r=0$ and $n=1$.
Using the same notation as in the proof of \RefThm{Thm:MCCSunbiased} we can write
\begin{align*}
	\mathbb{E}[\hat{g}_1^2]
	&= \int_{[0,1]^{m}} L_m^2 \max(f(\hat{\bsu}),0)^2 \,\, d\hat{\bsu}\\
	&= \int_{[0,1]^{m}} \chi^2(\bsu) \, L_m \max(f(\bsu),0)^2 \,\, d\bsu\\
	&\leq \int_{[0,1]^{m}}  \chi^2(\bsu)\max(f(\bsu),0)^2 \,\, d\bsu \\
	&= \mathbb{E}[g^2],
\end{align*}
where the inequality follows because $0\leq L_m$ and the last equality follows since $\chi^2(\bsu)=\chi(\bsu)$.
\end{proof}
While this conditional sampling scheme is very powerful, it cannot be used for methods such as LT.
Under this method, the multiplication by an othogonal matrix on $\bsz$ has the result that each component of $\tilde{\bsW}$ is, in general, given as a linear combination of the normal variates $z_1$ to $z_{mn}$.  
Therefore, one can not condition using the probabilities 
\begin{align*}
	&\mathbb{P}\left[ S_1(t_{j+1})<B | S_1(t_j) \right]
\end{align*}
since changing any variate $z_{j+1}$ such that $S_1(t_{j+1})<B$ will end up modifying $S_1(t_j)$ as well, again changing the probability
of knocking out conditional on the previous time step.
Therefore we propose an alternative conditional sampling scheme in the next section.

\begin{algorithm}
\caption{The Glasserman and Staum algorithm \cite{GS2001} for an up-\&-out option.}
\label{alg:gscs}
\begin{algorithmic}
\STATE Generate $N$ $mn$-dimensional uniform vectors $\bsu^{(i)}$, $i=1,\ldots,N$
\STATE Calculate the matrix $C$
\STATE Set $L=1$
\FOR{$i=1$ to $N$}
\FOR{$j=1$ to $m$}
\STATE $\hat{u}^{(i)}_{(j-1)n+1}=\mathbb{P}\left[ S_1(t_{j+1})<B | S_1(t_j) \right]u^{(i)}_{(j-1)n+1}$
\STATE Calculate $S_1(t_j),\ldots,S_n(t_j)$
\STATE Set $L=L \; \mathbb{P}\left[ S_1(t_{j+1})<B | S_1(t_j) \right]$
\ENDFOR
\STATE $g_i=L \max\{f(S_1(t_1),\ldots,S_n(t_m)),0\}$
\ENDFOR
\STATE $\hat{\mu}=\frac{1}{N}\sum_{i=1}^N g_i$ 
\STATE $\hat{\sigma}=\sqrt{\frac{1}{N(N-1)}\sum_{j=1}^N\left(g_i-\hat{\mu}\right)^2}$
\end{algorithmic}
\end{algorithm}

\section{Modified conditional sampling}\label{sect:LTCS}

We now derive an alternative conditional sampling scheme compatible with the LT algorithm discussed in \RefSec{sec:LTmethod}.
As an example we will again consider the up-\&-out option with payoff
\begin{align*}
 g(S_1(t_1),\ldots,S_n(t_m))
 &=
 \max\left(f(S_1(t_1),\ldots,S_n(t_m)),0\right)\,\mathbb{I}\left\{\max_jS_1(t_j)<B\right\}
 .
\end{align*}
For the option to stay alive we have the restriction
\begin{align}
 \sigma_1W_1(t_j) 
  &<
  \underbrace{\log\left(\frac{B}{S_1(0)}\right)-\left(r-\frac{\sigma_1^2}{2}\right)t_j}_{=b(t_j)}
  \qquad \text{for all }
  j=1,\ldots,m
  ,
  \label{eq:LTrestric}
\end{align}
where for a general covariance decomposition $A$
\begin{align*}
 \sigma_1W_1(t_j) 
  &=
  \left(A\Phi^{-1}(\bsu)\right)_j
\end{align*}
and the vector $\bsu$ consists of $mn$ uniform random variables.
The restriction\RefEq{eq:LTrestric} becomes
\begin{align*}
 \sum_{i=1}^{mn} a_{j,i} \Phi^{-1}(u_i)
  &<
  b(t_j)
  \qquad \text{for all }
  j=1,\ldots,m
  .
\end{align*}
We assume for the moment that $a_{j,1} > 0$ for $j=1,\ldots,m$, which is true if $c_{j,i} \ge 0$ for $j=1,\ldots,m$ and $i=1,\ldots,nm$, with $C$ the Cholesky factor in~\eqref{eq:cholfac}, and the sign of $Q_{\bullet,1}$ appropriately chosen (cf.\RefEqEq{eq:AisCQ}{eq:ITk1sol}).
Later we provide the solution for the general case.
Using this assumption we can modify the inequalities such that there is only one constraint on the first uniform variable assuming the others are given:
\begin{align*}
 \Phi^{-1}(u_1)
  &<
  \frac{b(t_j) - a_{j,2}\Phi^{-1}(u_2) - \ldots - a_{j,mn}\Phi^{-1}(u_{mn})}{a_{j,1}}
  \qquad \text{for all }
  j=1,\ldots,m
  .
\end{align*}
Equivalently we can write these $m$ conditions as the single condition
\begin{align*}
 u_1
  &<
  \underbrace{\Phi\left(\min_{j}\left[\frac{b(t_j) - a_{j,2}\Phi^{-1}(u_2) - \ldots - a_{j,mn}\Phi^{-1}(u_{mn})}{a_{j,1}}\right]\right)}_{=\Upsilon_u(u_2,\ldots,u_{mn})}
  .
\end{align*}
Our algorithm samples $u_1$ to $u_{mn}$, then calculates the upper bound $\Upsilon_u$ on $u_1$ using $u_2,\ldots,u_{mn}$, and then rescales $u_1$ to $\hat{u}_1=\Upsilon_u(u_2,\ldots,u_{mn})u_1$ in order to satisfy the barrier condition.
Our estimator is thus based on sampling
\begin{align*}
	\hat{g}_2
	&=
	\Upsilon_u(u_2,\ldots,u_{mn}) \max\left(f(\hat{u}_1,u_2,\ldots,u_{mn}),0\right)
	.
\end{align*}
For a down-\&-out option, a similar analysis can be done, leading to the condition
\begin{align*}
 u_1
  &>
  \underbrace{\Phi\left(\max_{j}\left[\frac{b(t_j) - a_{j,2}\Phi^{-1}(u_2) - \ldots - a_{j,mn}\Phi^{-1}(u_{mn})}{a_{j,1}}\right]\right)}_{=\Upsilon_d(u_2,\ldots,u_{mn})}
  .
\end{align*}
In case $C$ has negative elements, it is possible to have negative values for $a_{j,1}$.
Our algorithm can be modified to work around this problem:
suppose $a_{j,1}>0$ for $j\in\mathcal{P}\subseteq \{1,\ldots,m\}$ and $a_{j,1}<0$ for all other $j \notin \mathcal{P}$.
Then we find the conditions on $u_1$ as 
\begin{align*}
 \Phi^{-1}(u_1)
  &<
  \frac{b(t_j) - a_{j,2}\Phi^{-1}(u_2) - \ldots - a_{j,mn}\Phi^{-1}(u_{mn})}{a_{j,1}}
  \qquad \text{for all }
  j\in\mathcal{P}
\end{align*}
and
\begin{align*}
 \Phi^{-1}(u_1)
  &>
  \frac{b(t_j) - a_{j,2}\Phi^{-1}(u_2) - \ldots - a_{j,mn}\Phi^{-1}(u_{mn})}{a_{j,1}}
  \qquad \text{for all }
  j\notin\mathcal{P}
  .
\end{align*}
The restriction on $u_1$ now consists of an upper and a lower bound, denoted by $\Upsilon_u$ and $\Upsilon_d$ respectively, suppressing the dependence on $u_2$ to $u_{mn}$ for ease of notation.
The rescaled variable is then $\hat{u}_1=\Upsilon_d+(\Upsilon_u-\Upsilon_d)u_1$.
The estimator is then based on sampling
\begin{align*}
	\hat{g}_2
	&=
	\max(\Upsilon_u-\Upsilon_d,0)\max(f(\hat{u}_1,u_2,\ldots,u_{mn}),0)
	.
\end{align*}
We take the maximum of $\Upsilon_u-\Upsilon_d$ and $0$ as it could happen that $\Upsilon_u<\Upsilon_d$, and in this case the value of $\hat{g}_2$ should be taken as zero.
Examples of such payoffs are given in \RefSec{sec:dbba} and \RefSec{sec:sba2}.
An overview of this algorithm is presented in \RefAlg{alg:ltcs}.

To obtain a standard deviation we need to obtain a small number $M$ of randomized estimators which is standard practice in applying QMC.
\RefAlg{alg:ltcs} is generic in the sense that one can either use a lattice rule or lattice sequence with random shifting (see \cite{SJ94} for a general reference) or a digital net or digital sequence with digital shifting (see \cite{DP2010} for a general reference). 
For our numerical tests we will use either a digital sequence (being the Sobol' sequence with parameters from \cite{JK2008}) with digital shifting,  or a lattice sequence (with generating vector \textsf{exod8\_base2\_m13} from \cite{NUYWEB} constructed using the algorithm in~\cite{CKN2006}).

\begin{algorithm}
\caption{A modified conditional sampling algorithm}
\label{alg:ltcs}
\begin{algorithmic}
\STATE Generate $N$ $mn$-dimensional low-discrepancy points $\bsu^{(i)}$, $i=1,\ldots,N$
\STATE Generate $M$ random shifts $\bsDelta^{(i)}$, $i=1,\ldots,M$
\STATE Construct the matrix $A$ s.t. $AA'=\tilde{\Sigma}$
\FOR{$k=1$ to $M$}
\FOR{$i=1$ to $N$}
\STATE Randomize $\bsu^{(i)}$ using $\bsDelta^{(k)}$ and call this $\bsv^{(i,k)}$
\STATE Calculate the bounds $\Upsilon_d$ and $\Upsilon_u$ using $v_2^{(i,k)}$ to $v_{mn}^{(i,k)}$
\STATE $\hat{v}_1 = \Upsilon_d+(\Upsilon_u-\Upsilon_d) \, v_1^{(i,k)}$
\STATE $g_{k,i} = \max(\Upsilon_u-\Upsilon_d,0) \max\left(f(A\Phi^{-1}(\bsv^{(i,k)})),0\right)$
\ENDFOR
\ENDFOR
\STATE $\hat{\mu}=\frac{1}{MN}\sum_{k=1}^N\sum_{i=1}^M g_{k,i}$ 
\STATE $\hat{\sigma}=\sqrt{\frac{1}{M(M-1)}\sum_{i=1}^M\left(\frac{1}{N}\sum_{k=1}^N g_{k,i}-\hat{\mu}\right)^2}$
\end{algorithmic}
\end{algorithm}
So far we only considered barriers on one asset, whereas the payoff may depend on a basket. 
Our method can also deal with barriers on multiple stocks by calculating the conditions $\Upsilon_u$ and $\Upsilon_d$ for each asset individually, 
say $\Upsilon_d^{\ell}$ and $\Upsilon_u^{\ell}$ for asset $\ell$, and then taking $\max_{\ell}\Upsilon_d^{\ell}$ and $\min_{\ell}\Upsilon_u^{\ell}$ as the 
lower and upper bound respectively on $u_1$.
More complicated barrier conditions might require a root-finding method, see \RefSec{sec:root-finding}.

Similar as in the previous section we obtain the following results.
\begin{theorem}
	The estimator based on conditional sampling by
\begin{align}
	\hat{g}_2
	&=
	\max(\Upsilon_u-\Upsilon_d,0)\max(f(\hat{u}_1,u_2,\ldots,u_{mn}),0)
	,
	\label{Eq:estim2}
\end{align}
    where $\hat{u}_1=\Upsilon_d+(\Upsilon_u-\Upsilon_d)u_1$,
    is unbiased. That is,
	$\mathbb{E}[g] =\mathbb{E}[\hat{g}_2].$
	\label{Thm:LTCSunbiased}
\end{theorem}
\begin{proof}
Denote by $\chi$ the function
\begin{align*}
 \chi(\bsU)
 &=
 \begin{cases} 1 & \text{if $\max_{i\in\{1,\ldots,m\}} S(t_i)<B$,}
  \\
  0 &\text{otherwise.}
 \end{cases}
\end{align*}
Without loss of generality we assume $r=0$.
The value of the contract is then given as $\mathbb{E}\left[\chi(\bsU)\max(f(\bsU),0)\right]$ where we write out the dependencies explicitly in terms of $\bsU$.
We now want to prove that this equals $\mathbb{E}\left[\max(\Upsilon_u-\Upsilon_d,0)\max(f(\hat{\bsU}),0)\right]$. 
\begin{align*}
 \mathbb{E}\left[\chi(\bsU) \max(f(\bsU),0)\right]
 &= \int_{[0,1]^{mn}} \chi(\bsu) \max(f(u_1,\ldots,u_{mn}),0) \,\, d\bsu \\
 &= \int_{[0,1]^{mn-1}} \int_{[\Upsilon_d,\Upsilon_u]} \max(f(u_1,\ldots,u_{mn}),0) \,\, du_1\cdots du_{mn}
 .
\end{align*}
By using the change of variables $(u_1-\Upsilon_d)/(\Upsilon_u-\Upsilon_d)=\hat{u}_1$ and $u_i=\hat{u_i}$ for $i=2,\ldots,mn$ we find (using the notation $\hat{\bsu}=(\hat{u}_1,u_2,\ldots,u_{mn})$) that
\begin{align*}
 &\mathbb{E}\left[\chi(\bsU)\max(f(\bsU),0)\right]
 \\
 &\qquad= \int_{[0,1]^{mn-1}} \int_{[0,1]} \max(\Upsilon_u-\Upsilon_d,0)\max(f(\Upsilon_d+(\Upsilon_u-\Upsilon_d)\hat{u}_1,\ldots,u_{mn}),0) \,\, d\hat{u}_1\cdots du_{mn}
 \\
 &\qquad= \int_{[0,1]^{mn}} \max(\Upsilon_u-\Upsilon_d,0)\max(f(\Upsilon_d+(\Upsilon_u-\Upsilon_d)\hat{u}_1,\ldots,u_{mn}),0) \,\, d\hat{\bsu}
 \\
 &\qquad= \mathbb{E}\left[\max(\Upsilon_u-\Upsilon_d,0)\max(f(\hat{\bsU}),0)\right]
 ,
\end{align*}
which is what we needed to prove.
\end{proof}

\begin{theorem}
	When using regular Monte Carlo, the estimator defined in\RefEq{Eq:estim2} for an up-\&-out option has reduced variance. That is,
	$\var\left[ \hat{g}_2\right] \leq \var\left[g \right].$
	The inequality is strict if $\mathbb{P}\left[\max_jS_1(t_j)\geq B\right]>0$ and $\mathbb{E}\left[g\right]>0$, i.e., if there is any chance of knock-out and positive payoff.
\end{theorem}
\begin{proof}
Because of \RefThm{Thm:LTCSunbiased} we only need to prove that $\mathbb{E}[\hat{g}_2^2]\leq\mathbb{E}[g_2^2]$.
Using the same notation as in the proof of \RefThm{Thm:LTCSunbiased} we can write
\begin{align*}
	\mathbb{E}[\hat{g}_2^2]
	&= \int_{[0,1]^{mn}} \max(\Upsilon_u-\Upsilon_d,0)^2 \max(f(\Upsilon_d+(\Upsilon_u-\Upsilon_d)\hat{u}_1,\ldots,u_{mn}),0)^2 \,\, d\hat{u}_1\cdots du_{mn}\\
	&= \int_{[0,1]^{mn}} \chi^2(\bsU)\max(\Upsilon_u-\Upsilon_d,0) \max(f(u_1,\ldots,u_{mn}),0)^2 \,\, du_1\cdots du_{mn}\\
	&\leq \int_{[0,1]^{mn}}  \chi^2(\bsU)\max(f(u_1,\ldots,u_{mn}),0)^2 \,\, du_1\cdots du_{mn} \\
	&= \mathbb{E}[g_2^2],
\end{align*}
where the inequality follows because $\max(\Upsilon_u - \Upsilon_d, 0) \leq 1$ and the last equality since $\chi^2(\bsU)=\chi(\bsU)$.
\end{proof}

For our case, when using quasi-Monte Carlo, we have the following theorem on variance.
\begin{theorem}
	When using a randomly shifted quasi-Monte Carlo method, the estimator defined in\RefEq{Eq:estim2} for an up-\&-out option has reduced variance. That is,
	$\var[ \hat{g}_2] \leq \var[g ].$
	The inequality is strict if $\mathbb{P}\left[\max_jS_1(t_j)\geq B\right]>0$ and $\mathbb{E}\left[f\right]>0$, i.e., if there is any chance of knock-out and positive payoff. 
	\label{Thm:RedVar}
\end{theorem}
\begin{proof}
When using lattice rules or lattice sequences the standard randomization technique is to use random shifting \cite{SJ94}.
For digital nets and digital sequences a standard technique is to use random digital shifting \cite{DP2010}.
Results on variance for both these methods can be found in \cite{LL2000} and \cite{LL2002b}.
As random shifting is easier to understand we will use this in the proof, but the proof can be trivially modified for random digital shifting.
In what follows we will use the notation $\{x\}=x-\lfloor x\rfloor$ to denote the fractional part of $x$.
When using $M \ge 1$ random shifts the estimator becomes
\begin{align*} 
  \frac{1}{M} \sum_{k=1}^M
  \frac{1}{N} \sum_{i=1}^N
  \max(\Upsilon_u-\Upsilon_d,0) 
  \max\left(f\left(\Upsilon_d+(\Upsilon_u-\Upsilon_d) v^{(i,k)}_1, v^{(i,k)}_2, \ldots, v^{(i,k)}_{mn}\right), 0 \right),
\end{align*}
where the bounds $\Upsilon_u$ and $\Upsilon_d$ are calculated based upon the shifted points and these shifted points are obtained as
\begin{align*}
  v^{(i,k)}_j
  &=
  \{ u^{(i)}_j + \Delta^{(k)}_j \}
  &\text{for } j=1,\ldots,nm,
\end{align*}
with the random shift vectors $\bsDelta^{(j)}$ i.i.d.\ uniform variables over $[0,1]^{nm}$. 
For ease of notation, write
\begin{align*}
	F(\{\bsu^{(i)}+\bsDelta^{(k)}\})
	&=
	f\left(\Upsilon_d+(\Upsilon_u-\Upsilon_d) v^{(i,k)}_1, v^{(i,k)}_2, \ldots, v^{(i,k)}_{mn}\right)
	.
\end{align*}
We will keep suppressing the arguments of the bounds $\Upsilon_d$ and $\Upsilon_u$ for ease of notation, but as a reminder we note that these depend on the integration variables as well.  
By observing that the integration of the linear map $\bsDelta \mapsto \{\bsu+\bsDelta\}$ over $[0,1)^{mn}$ is the same as integrating $\bsDelta$ itself over $[0,1)^{mn}$, we can easily prove that the above estimator is equivalent to $\mathbb{E}[\hat{g}_2]$. Indeed,
\begin{align*} 
	&
	\mathbb{E}_{\Delta}\left[\frac{1}{M}\sum_{k=1}^M\frac{1}{N}\sum_{i=1}^N\max(\Upsilon_u-\Upsilon_d,0)\max(F(\{\bsu^{(i)} +\bsDelta^{(k)}\}),0)\right] \\
	&\qquad=
	\frac{1}{M}\sum_{k=1}^M\frac{1}{N}\sum_{i=1}^N\int_{[0,1)^{mn}} \max(\Upsilon_u-\Upsilon_d,0)\max(F(\{\bsu^{(i)} +\bsDelta^{(k)} \}),0) \, d\bsDelta^{(k)} \\
	&\qquad=
	\frac{1}{M}\sum_{k=1}^M\frac{1}{N}\sum_{i=1}^N\int_{[0,1)^{mn}} \max(\Upsilon_u-\Upsilon_d,0)\max(F(\bsDelta^{(k)} ),0) \, d\bsDelta^{(j)}  \\
	&\qquad=
	\frac{1}{M}\sum_{k=1}^M\frac{1}{N}\sum_{i=1}^N \mathbb{E}[\hat{g}_2] \\
	&\qquad=
	\mathbb{E}[\hat{g}_2].
\end{align*} 
Similarly, we find for the second moment
\begin{align*} 
	&
	\mathbb{E}_{\Delta}\left[\left(\frac{1}{M}\sum_{k=1}^M\frac{1}{N}\sum_{i=1}^N\max(\Upsilon_u-\Upsilon_d,0)\max(F(\{\bsu^{(i)} +\bsDelta^{(k)}\}),0)\right)^2\right] \\
	&\qquad=
	\int_{[0,1)^{mn}}\cdots\int_{[0,1)^{mn}} \left(\frac{1}{M}\sum_{k=1}^M\frac{1}{N}\sum_{i=1}^N\max(\Upsilon_u-\Upsilon_d,0)\max(F(\{\bsu^{(i)} +\bsDelta^{(k)}\}),0)\right)^2 d\bsDelta^{(1)}\cdots d\bsDelta^{(M)} \\
	&\qquad=
	\int_{[0,1)^{mn}}\cdots\int_{[0,1)^{mn}} \left(\frac{1}{M}\sum_{k=1}^M\frac{1}{N}\sum_{i=1}^N\max(\Upsilon_u-\Upsilon_d,0)\max(F( \bsDelta^{(k)}),0)\right)^2 d\bsDelta^{(1)}\cdots d\bsDelta^{(M)} \\
	&\qquad=
	\frac{1}{M^2 N^2}\sum_{j=1}^{M^2N^2}  \int_{[0,1)^{mn}} \max(\Upsilon_u-\Upsilon_d,0)^2\max(F( \bsDelta),0)^2 \, d \bsDelta \\
	&\qquad=
	\frac{1}{M^2 N^2}\sum_{j=1}^{M^2N^2}  \int_{[0,1)^{mn}} \max(\Upsilon_u-\Upsilon_d,0) \chi^2(\bsDelta)\max(f(\bsDelta),0)^2 \, d\bsDelta \\
	&\qquad\leq
	\frac{1}{M^2 N^2}\sum_{j=1}^{M^2N^2}  \int_{[0,1)^{mn}} \chi^2(\bsDelta)\max(f(\bsDelta),0)^2 \, d\bsDelta,
\end{align*} 
which concludes the proof.
\end{proof}

\begin{remark}
We only show that our conditional sampling scheme has reduced variance compared to the unconditional method.
We do not show that it has reduced variance compared to the Glasserman \& Staum conditional sampling scheme.
In fact, this statement will not always be true, as will be illustrated by a numerical example in the next section.
The scenario in which this happens however is one where an asset is forced to stay within two very tight barriers, which is an unrealistic choice of parameters.
We are not aware of a theoretical quantification when either method will outperform the other one.

Heuristically, what our conditional scheme really does is shifting the asset paths entirely to satisfy the barrier condition.
Assuming no mixed signs in the first column of $A$, this will never be a problem for a single barrier.
When the signs are mixed or when there are two barriers, it can occur that the path can not be shifted without hitting a barrier.
In fact, looking at the definition of $\Upsilon_d$ and $\Upsilon_u$, in this case we can always choose a barrier such that $\Upsilon_d>\Upsilon_u$.
The Glasserman \& Staum scheme on the other hand samples the asset path conditionally in each time step using a different random variate. 
This is not possible for a general path construction, like the LT method, since all the random variates influence the asset in each time step,
see also the discussion at the end of \RefSec{sec:CondSampl}.
The numerical results show that the inability to find suitable bounds on $z_1$ does not necessarily mean our method performs worse than the Glasserman \& Staum conditional sampling scheme. 
\end{remark}


\section{Numerical Results}\label{sec:numres}

In the numerical examples, unless specified otherwise, we use the Sobol' sequence with parameters from \cite{JK2008} and digital shifting \cite{DP2010} and we calculate all the columns of $Q$ for the LT construction.
We use this same setup also for the examples of the two extensions in \RefSecSec{sec:knock-in}{sec:root-finding}.
Other QMC point sets can be used as well.
In \RefSec{ex:sbba} we used both the Sobol' sequence and a lattice sequence to illustrate that the choice of QMC point set is more or less arbitrary (as long as it is of good quality).
For more information on lattice rules, see \cite{CKN2006,NUYWEB}. 

\subsection{Single barrier Asian basket}\label{ex:sbba}
Consider an Asian barrier option on four assets:
\begin{align*}
 g 
  &= 
  \max\left(\frac{1}{4\times130}\sum_{i=1}^{4}\sum_{j=1}^{130} S_i(t_j)-K,0\right)\mathbb{I}\left\{\max_{j=1,\ldots,130} S_1(t_j)< B\right\}
  .
\end{align*}
We will consider the valuation of this option under several model parameters.
The fixed parameters are $S_i(0)=100$ for $i=1,\ldots,4$, $\sigma_2=\sigma_3=25\%$, $\sigma_4=35\%$, $r=5\%$ and $T=6$ months.
Here we have taken $m=130$.
We consider two correlation matrices:
\begin{align*}
	P_1 
	&= 
	\begin{pmatrix} 1 & 0.6 & 0.6 & 0.6 \\ 0.6 & 1 & 0.6 & 0.6 \\ 0.6 & 0.6 & 1 & 0.6 \\ 0.6 & 0.6 & 0.6 & 1 \end{pmatrix}
	&\text{and}&&
	P_2 
	&= 
	\begin{pmatrix} 1 & -0.5 & 0.6 & 0.2 \\ -0.5 & 1 & -0.2 & -0.1 \\ 0.6 & -0.2 & 1 & 0.25 \\ 0.2 & -0.1 & 0.25 & 1 \end{pmatrix}
	.
\end{align*}
In what follows, we denote the LT method using a low-discrepancy point set by QMC+LT, and the conditional sampling scheme we derived for the LT method using such a point set by QMC+LT+CS. 
The conditional sampling scheme for Monte Carlo will be denoted by MC+CS. 
We report the improvement of the standard deviations of the value estimates of these methods to that of the conditional sampling scheme for regular Monte Carlo (MC+CS) for various choices of 
$\sigma_1$, $K$ and $B$ in \RefTab{table:Ex4} using both the Sobol' sequence and a lattice sequence. 
We see similar results for both QMC point generators.

The first example takes $\sigma_1=25\%$, $K=70$ and the extreme choice of $B=10000$. Clearly, the barrier will almost never be hit, and we should obtain the same result for QMC+LT as for QMC+LT+CS. 
This example also shows the power of the QMC+LT method compared to using MC+CS, as the standard deviation under the former method is about 26 times smaller. 

We then test three different cases for $\sigma_1=25\%$ and again three different cases for $\sigma_1=55\%$.
In all of these cases the condition of hitting the barrier comes into play and we can see the advantage of conditional sampling clearly.
The QMC+LT method performs worse than MC+CS on one occasion but delivers good results for all other cases.
Our new QMC+LT+CS method always outperforms the QMC+LT method and outperforms the MC+CS method by factors ranging from $1.6$ to $8$ with respect to the standard deviation.
%
%
%
%
%

In \RefFig{fig:example5} we show the convergence graphs for the second and fifth examples. 
The convergence factors $\alpha$, which we obtained from the linear regression $\log(\sigma)=\beta-\alpha\log(N)$, are shown in \RefTab{tab:ex5conv}. We see that in two out of four cases our conditional scheme improves the convergence, and in the other cases the convergence is about the same as the original QMC+LT method. 
Overall, the results of our conditional sampling scheme QMC+LT+CS are very satisfactory. 

\begin{table}
\begin{center}
\begin{tabular}{ cA c|cA c}
\multicolumn{6}{c}{Sobol' sequence}\\
\hline
$(P, \sigma_1, B, K)$ & QMC+LT+CS & QMC+LT & $(P, \sigma_1, B, K)$ & QMC+LT+CS & QMC+LT  \\
\hline  
$(P_1,0.25,10000,70)$ & $2682$\% & $2682$\% &
$(P_2,0.25,10000,70)$ & $2575$\% & $2575$\%  \\[1mm]
$(P_1,0.25,125,70)$ & $786$\% & $154$\% &
$(P_2,0.25,125,70)$ & $621$\% & $197$\% \\
$(P_1,0.25,105,70)$ & $368$\% & $128$\% &
$(P_2,0.25,105,70)$ & $414$\% & $170$\% \\
$(P_1,0.25,110,100)$ & $231$\% & $120$\% &
$(P_2,0.25,110,100)$ & $304$\% & $165$\% \\[1mm]
$(P_1,0.55,105,70)$ & $287$\% & $110$\% &
$(P_2,0.55,105,70)$ & $234$\% & $144$\% \\
$(P_1,0.55,105,90)$ & $190$\% & $88$\% &
$(P_2,0.55,105,90)$ & $157$\% & $105$\% \\
$(P_1,0.55,150,110)$ & $278$\% & $135$\% &
$(P_2,0.55,150,110)$ & $269$\% & $126$\%  \\
\multicolumn{6}{c}{ }\\
\multicolumn{6}{c}{Lattice sequence}\\
\hline
$(P, \sigma_1, B, K)$ & QMC+LT+CS & QMC+LT & $(P, \sigma_1, B, K)$ & QMC+LT+CS & QMC+LT  \\
\hline  
$(P_1,0.25,10000,70)$ & $3457$\% & $3457$\% & $(P_2,0.25,10000,70)$ & $3318$\% & $3318$\% \\[1mm] 
$(P_1,0.25,125,70)$ & $536$\% & $143$\% & $(P_2,0.25,125,70)$ & $594$\% & $173$\% \\ 
$(P_1,0.25,105,70)$ & $358$\% & $144$\% & $(P_2,0.25,105,70)$ & $319$\% & $149$\% \\ 
$(P_1,0.25,110,100)$ & $190$\% & $121$\% & $(P_2,0.25,110,100)$ & $255$\% & $187$\% \\ 
$(P_1,0.55,105,70)$ & $250$\% & $124$\% & $(P_2,0.55,105,70)$ & $273$\% & $110$\% \\[1mm] 
$(P_1,0.55,105,90)$ & $150$\% & $103$\% & $(P_2,0.55,105,90)$ & $177$\% & $105$\% \\ 
$(P_1,0.55,150,110)$ & $341$\% & $209$\% & $(P_2,0.55,150,110)$ & $267$\% & $139$\% \\                        
\end{tabular}
\caption{Single barrier Asian basket. The reported numbers are the standard deviations of the MC+CS method divided by those of the QMC+LT and QMC+LT+CS methods. The MC+CS method uses $163840$ samples, while the QMC+LT and QMC+LT+CS methods use $4096$ samples and $40$ independent shifts.
	 }
\label{table:Ex4}
\end{center}
\end{table}

\begin{figure}[ht]
\centering
\includegraphics[scale=0.5]{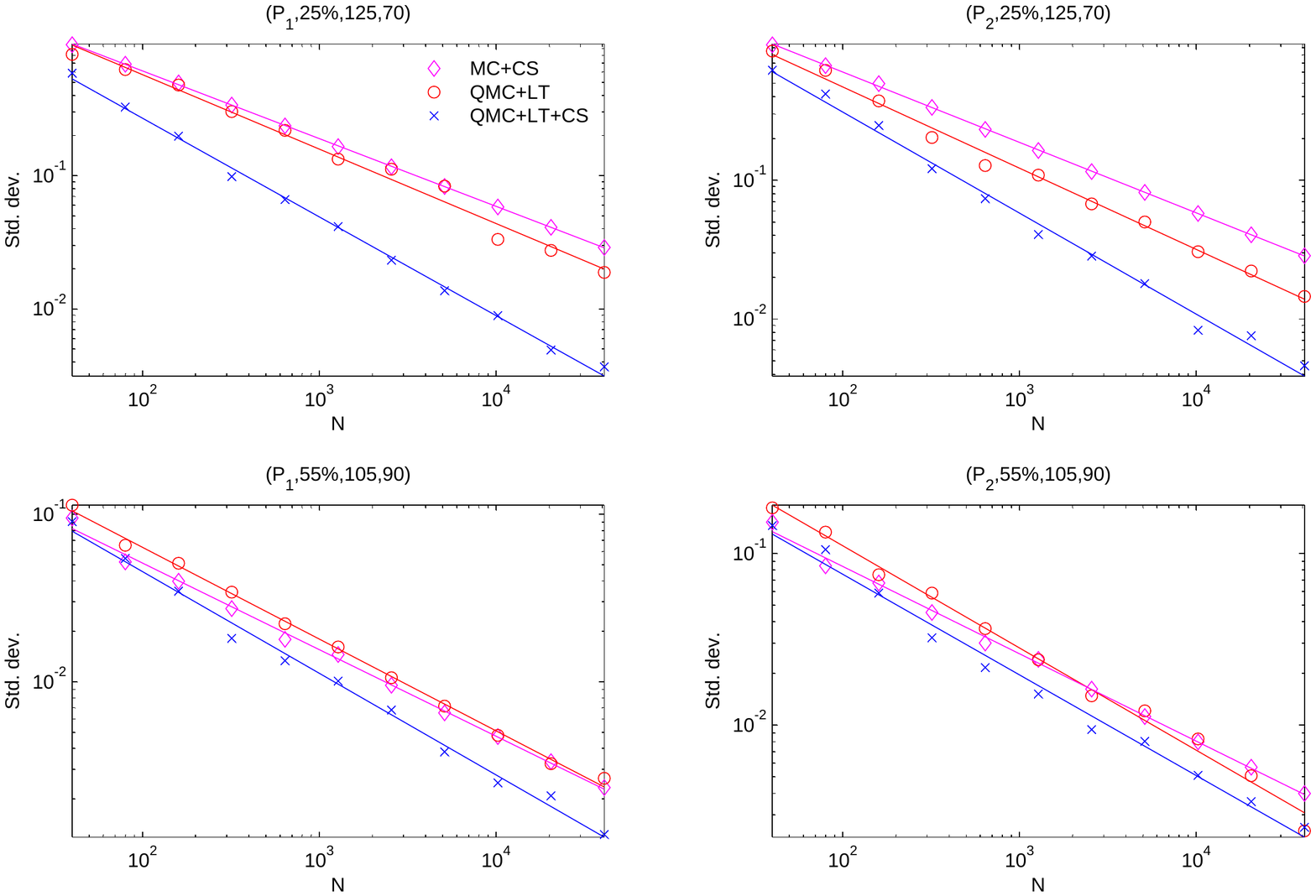}
\caption{Convergence graphs corresponding to the single barrier Asian basket. The title of each graph denotes $(P,\sigma_1,B,K)$.}\label{fig:example5}

\end{figure}

\begin{table}
\begin{center}
\begin{tabular}{ cA c c}
     $(P, \sigma_1, B, K)$ & QMC+LT+CS & QMC+LT & MC+CS  \\
\hline  
$(P_1,0.25,125,70)$ & $0.6991$ & $0.5363$ & $0.4939$ \\
$(P_2,0.25,125,70)$ & $0.7296$ & $0.5321$ & $0.4932$ \\
$(P_1,0.55,105,70)$ & $0.5526$ & $0.5838$ & $0.5046$ \\
$(P_2,0.55,105,70)$ & $0.6145$ & $0.5828$ & $0.5043$ \\                           
\end{tabular}
\caption{The convergence factors $\alpha$ obtained from the linear regression $\log(\sigma)=\beta-\alpha\log(N)$ for single barrier Asian basket.}
\label{tab:ex5conv}
\end{center}
\end{table}

\subsection{Double barrier binary Asian}\label{sec:dbba}
Consider the following option payoff on a single asset:
\begin{align*}
 g
  &= 
 \mathbb{I}\left\{\frac{1}{m}\sum_{i=1}^{m}S(t_i)\geq 100\right\}\mathbb{I}\left\{\min_{i=1,\ldots,m} S(t_i)\geq B^L\right\}\mathbb{I}\left\{\max_{i=1,\ldots,m} S(t_i)\leq B^U\right\}
  .
\end{align*}
The model parameters are $S(0)=100$, $\sigma=30\%$, $T=3$ months and $r=0\%$.
The results are presented in \RefTab{table:Ex2}.

\begin{table}[h]
\begin{center}
\begin{tabular}{c cA c}
     $(B^L,B^U)$ & $m$ & QMC+LT+CS & QMC+LT \\
\hline 
$(1,1000)$ & $60$ & $1685$\% & $1685$\% \\
$(50,150)$ & $60$ & $1692$\% & $1180$\% \\
$(90,110)$ & $60$ & $292$\% & $118$\% \\[1mm]
$(98,102)$ & $2$ & $296$\% & $26$\% \\
$(98,102)$ & $3$ & $100$\% & $13$\% \\
$(98,102)$ & $4$ & $40$\% & $4$\%
\end{tabular}
\caption{Double barrier binary Asian. The reported numbers are the standard deviations of the MC+CS method divided by those of the QMC+LT and QMC+LT+CS methods. The MC+CS method uses $163840$ samples, while the QMC+LT and QMC+LT+CS methods use $4096$ samples and $40$ independent shifts.
	 }
\label{table:Ex2}
\end{center}
\end{table}

In the first three examples, the number of observation dates $m$ is fixed at $60$, and the barriers are chosen increasingly closer to $S(0)$. 
As the barriers are tightened the performance of both the QMC+LT and QMC+LT+CS method drop, but the QMC+LT+CS scheme clearly keeps the upper hand.
In order to push the QMC+LT+CS method to the limit, we consider the extreme example of unrealisticly tight barriers $(B^L,B^U)=(98,102)$ and choose $m=2,3,4$. 
These results are shown in the last three rows of \RefTab{table:Ex2}.
Although our method gets a serious variance reduction compared to the original QMC+LT method, we also notice that for $m=4$ our QMC+LT+CS scheme does not outperform the MC+CS method. 
To understand what is happening, we can plot the projection of the payoff on the first two dimensions $u_1$ and $u_2$.
This is shown in \RefFig{fig:Ex2m}.
We see that the LT algorithm is not able to separate the positive payoffs from the zero payoffs for increasing $m$ (where, e.g., in \cite{NW2012} it is observed that having a good separation will make the QMC method more effective).
A possible solution here might be to use a non-linear method.
Another observation is that in case of a tight double barrier, it is not always possible to find a valid bound on $u_1$.
This is illustrated in \RefFig{fig:Ex2db} where the path is drawn for several of $u_1$ for $m=4$ for fixed values of $u_2,\ldots,u_4$.
The thicker straight lines indicate the barriers at $98$ and $102$.
Clearly, there is no possible choice for $u_1$ such that the asset path stays between these barriers. 
In view of\RefEq{Eq:estim}, this means that this sample will be set to zero, as $\max(\Upsilon_u-\Upsilon_d,0)=0$.
While this reduces the efficiency of our method, it must be noted that a product with such extreme barriers is not encountered in the market.
We also want to stress that our estimator will still be unbiased, and as the next example will show, this observation does not necessarily imply that our method underperforms the MC+CS method in these cases. 

\begin{figure}[ht]
\centering
\subfigure[$m=2$]{
\includegraphics[scale=0.4]{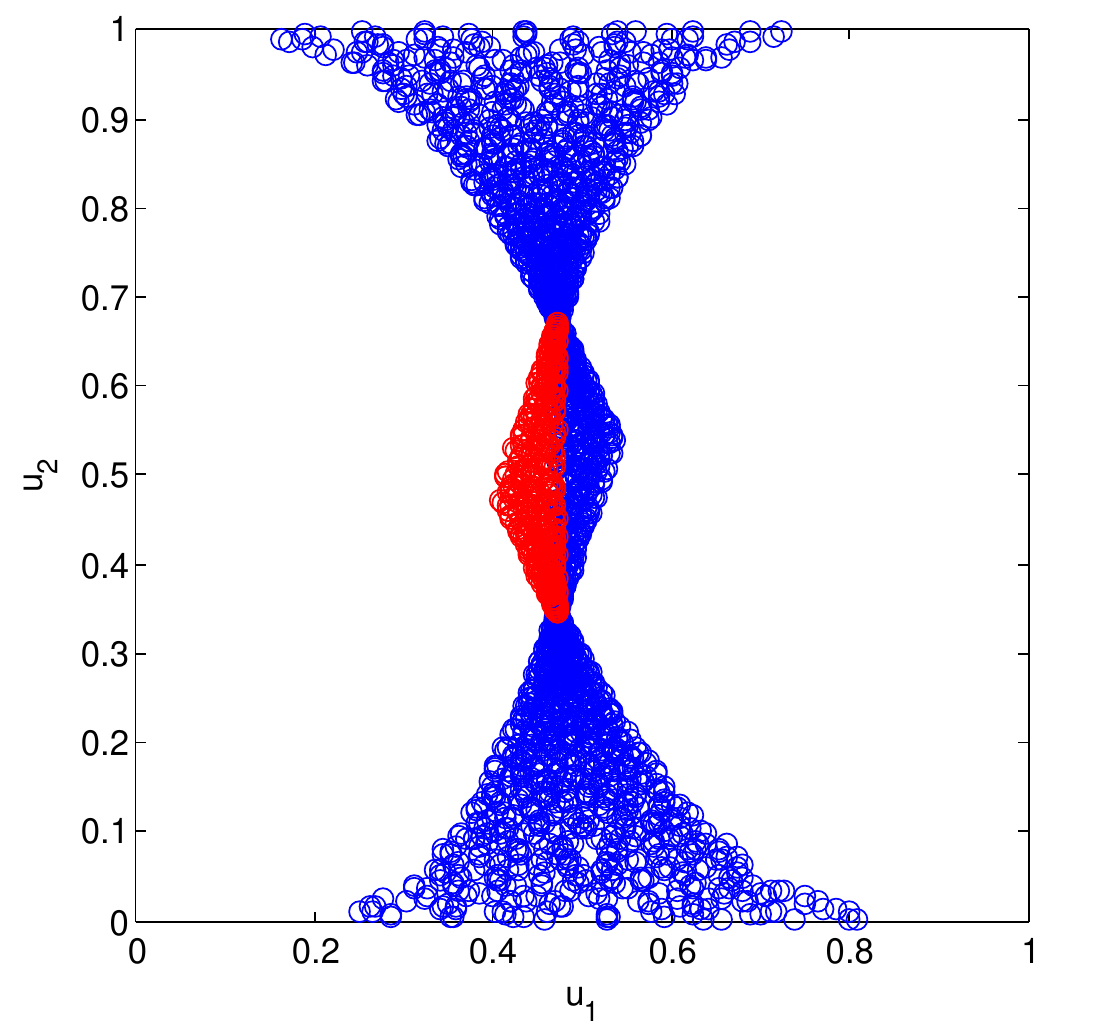}
\label{fig:m2ud}
}
\subfigure[$m=3$]{
\includegraphics[scale=0.4]{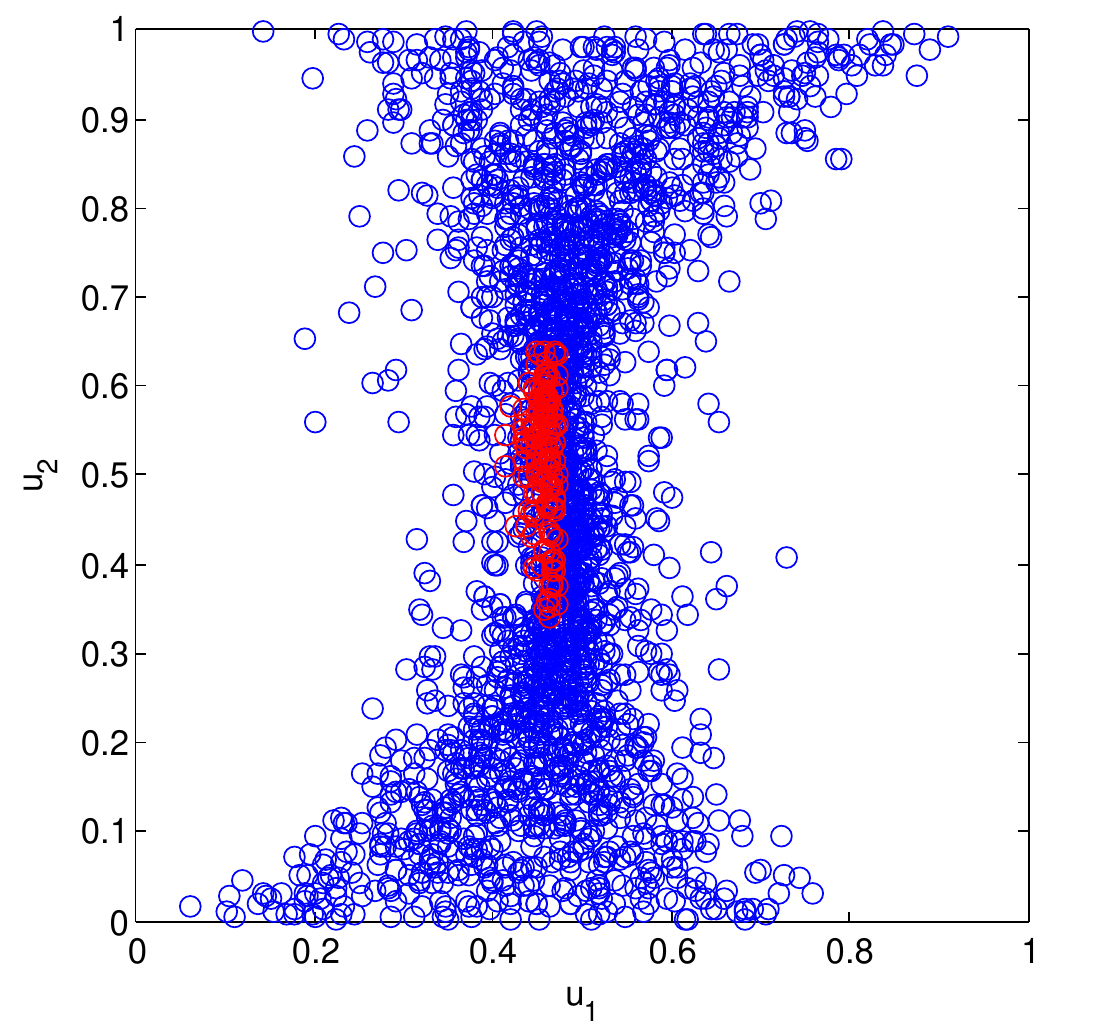}
\label{fig:m3ud}
}
\subfigure[$m=4$]{
\includegraphics[scale=0.4]{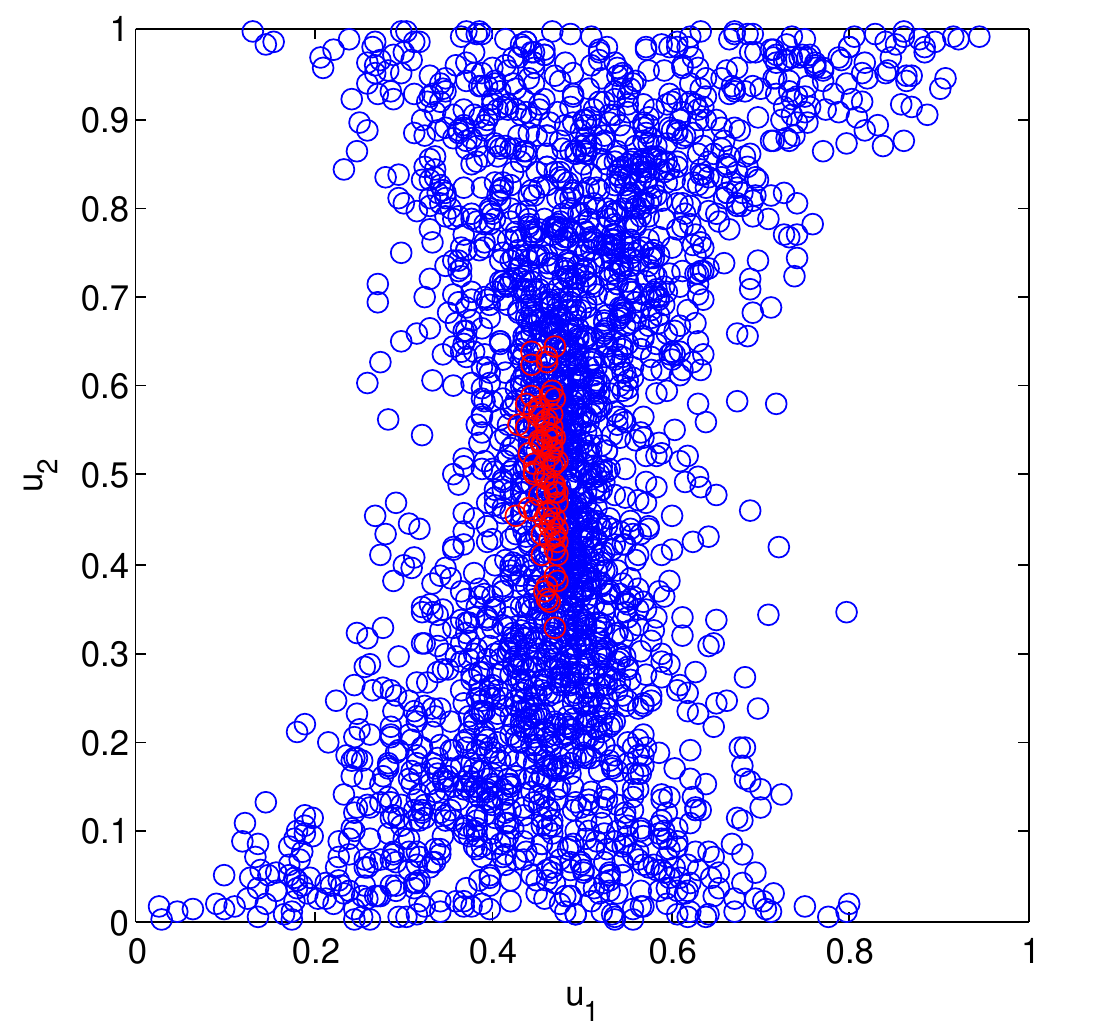}
\label{fig:m16ud}
}
\caption{The projection of the double barrier binary Asian payoff on the first two dimensions $u_1$ and $u_2$ under the QMC+LT+CS method. Blue circles indicate a zero payoff, red circles a non-zero payoff.}\label{fig:Ex2m}
\end{figure}

\begin{figure}[ht]
\centering
\includegraphics[scale=0.5]{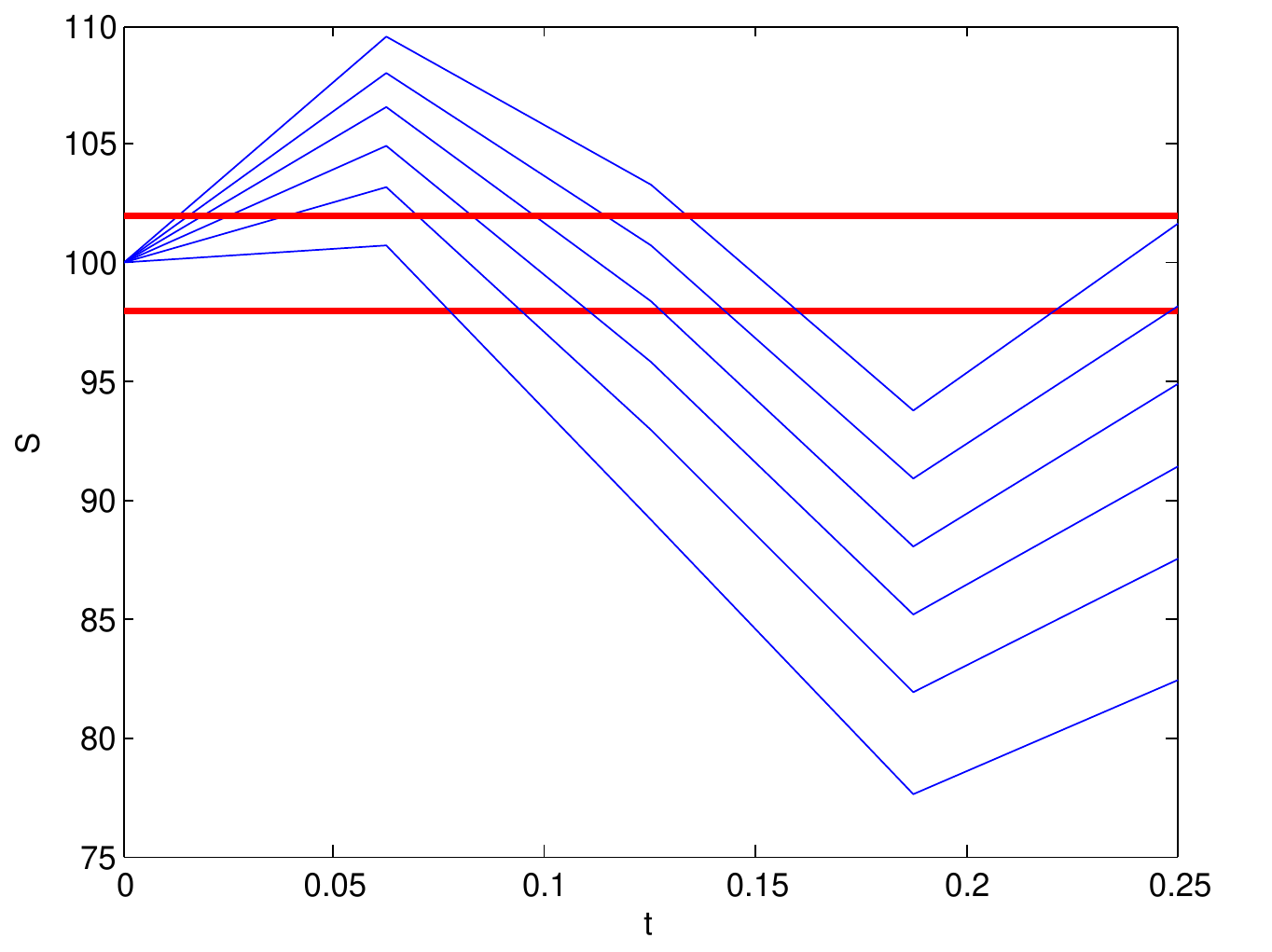}
\caption{The effect of $u_1$ on a sample path for the double barrier binary Asian with very narrow barriers, when $u_2, u_3$ and $u_4$ are fixed.}\label{fig:Ex2db}
\end{figure}

\subsection{Single barrier Asian}\label{sec:sba2}

We now construct an example for which the first $m$ elements of the first column of $A$ do not have the same sign and we thus need both $\Upsilon_u$ and $\Upsilon_d$ for the QMC+LT+CS algorithm.
In this example we take the following payoff on two assets:
\begin{align*}
  g
  &= 
 \left(\frac{1}{4}\left( S_1(t_1)+S_2(t_1)+S_1(t_2)+S_2(t_2)\right)\geq 1\right)\mathbb{I}\left\{\min\left( S_1(t_1),S_1(t_2)\right)\leq B\right\}
  .
\end{align*}
The model parameters are $\sigma_1=40\%$, $\sigma_2=60\%$, $r=8\%$, $S_1(0)=S_2(0)=1$, $t_1=\frac{1}{2}$, $t_2=1$ and $B=1.1$.
Straightforward calculations show that $a_{11}<0$ and $a_{21}>0$ if
\begin{align*}
	&\rho\in\left[ \frac{-\sigma_1}{\sigma_2e^{r-\sigma^2_2/2}}, \frac{-2\sigma_1}{\sigma_2(e^{(r-\sigma^2_2/2)/2} + e^{r-\sigma^2_2/2})}\right] \approx [-73.68\%,-71.84\%].
\end{align*}
We take $\rho=-72\%$.
In this case, it is possible that $\Upsilon_d>\Upsilon_u$ and our sample is wasted.
This happened in our simulation about $50\%$ of the time. 
Still we get a very nice result as the standard deviation of the MC+CS method divided by those of the QMC+LT and QMC+LT+CS methods based on $4096$ samples and $40$ shifts result in $384\%$ and $399\%$ respectively. 
A convergence plot is given in figure \RefFig{fig:Example3:convergence} while in \RefFig{fig:Example3:u1effect} the effect of $u_1$ on a sample path for which $\Upsilon_d>\Upsilon_u$ is given. 
Here the asset path is drawn for several values of $u_1$ and a fixed value of $u_2$.
The thick straight line indicates the barrier at $1.1$.
Clearly, there is no possible choice for $u_1$ such that the asset path stays below the barrier. 
We again want to stress that this does not imply unbiasedness, nor underperformance of our method, as the standard deviation is still about four times smaller than that of the MC+CS method.

\begin{figure}[ht]
\centering
\subfigure[Convergence graph.]{\label{fig:Example3:convergence}
\includegraphics[scale=0.5]{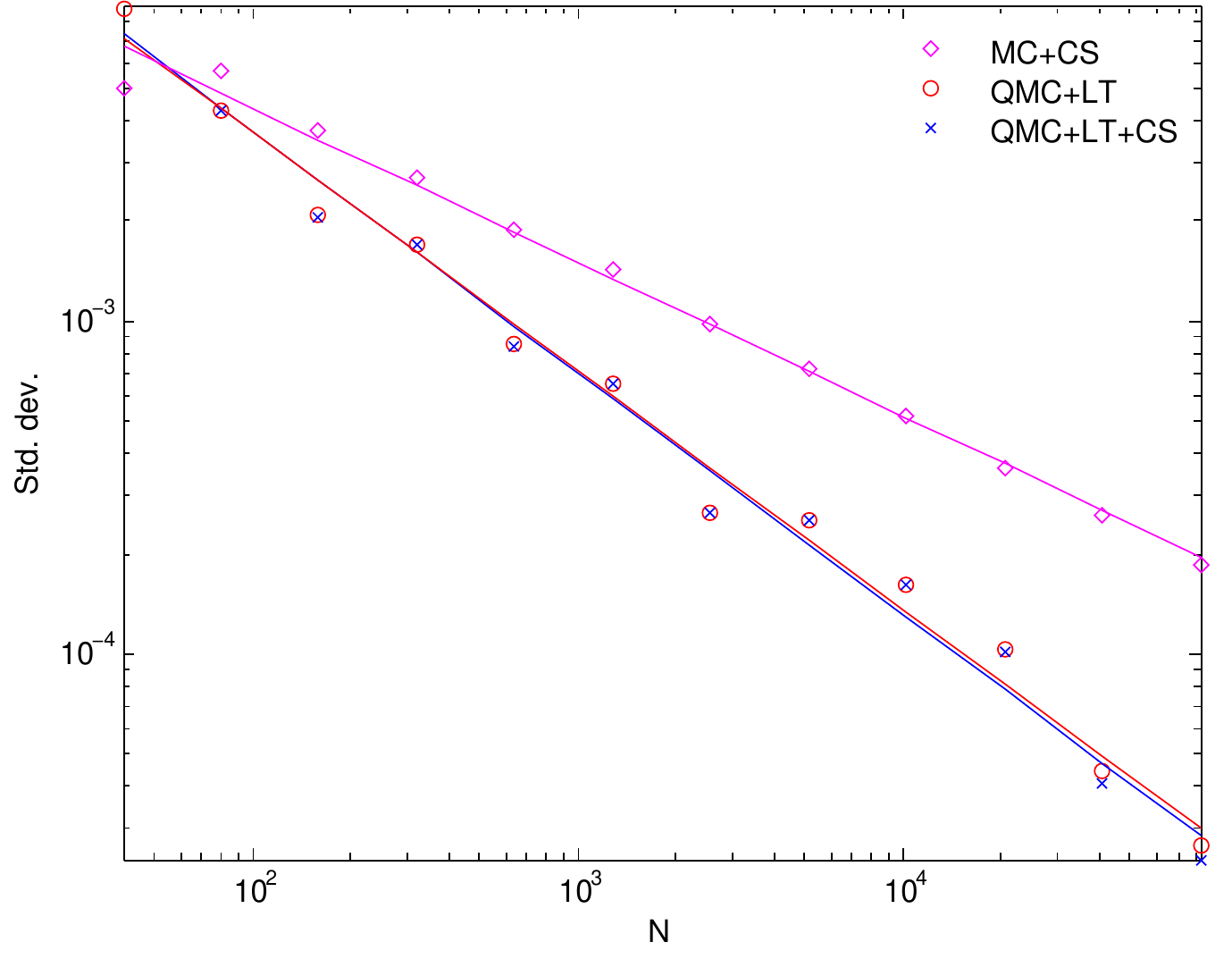}
}
\subfigure[Effect of $u_1$ on a sample path when $u_2$ is fixed.]{\label{fig:Example3:u1effect}
\includegraphics[scale=0.5]{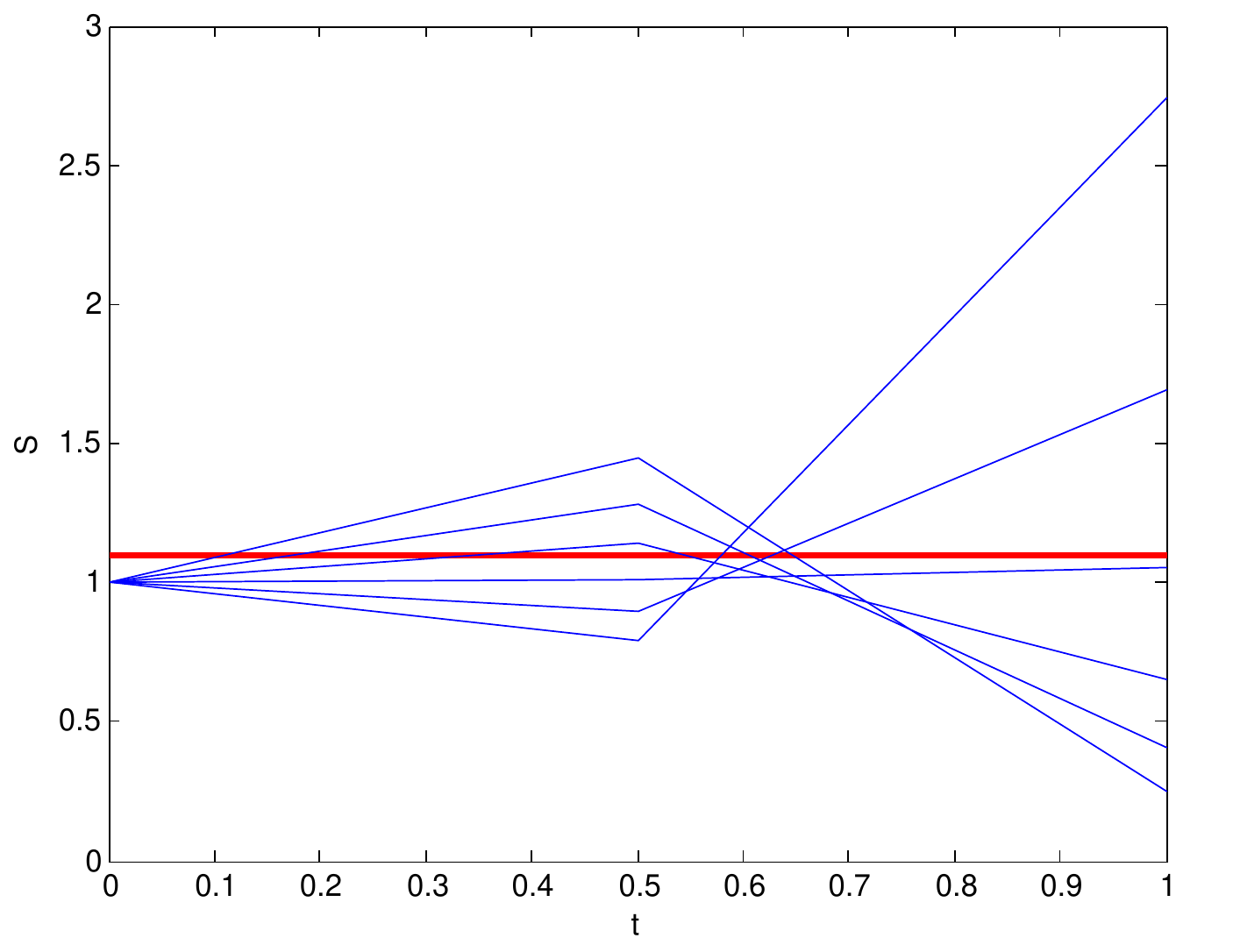}
}
\caption{Single barrier Asian example with mixed signs for the conditional sampling.}\label{fig:Example3}
\end{figure}

\subsection{Binary barrier}
As a final experiment consider the following option payoff on one asset:
\begin{align*}
 g
  &= 
  \mathbb{I}\left\{\max_i S(t_i)\leq 105\right\}
  .
\end{align*}
The model parameters are $S(0)=100$, $\sigma=30\%$, $T=3$ months and $r=0\%$.
We are essentially integrating a closed volume in the $m$-dimensional hypercube.
The resulting convergence factors are given in \RefTab{table:conv}.
These factors are calculated using the linear regression $\log(\sigma)=\beta-\alpha\log(N)$.
As we can see, increasing $m$ decreases the convergence speed for both QMC+LT and QMC+LT+CS.
This example illustrates that conditional sampling results in a variance reduction, but not necessarily in an increase in convergence speed. 
One possible explanation is as given in \RefSec{sec:dbba}, i.e., due to the barrier condition the LT algorithm is not capable any more of aligning the discontinuities with the axes.
Furthermore, as illustrated in \cite{NW2012}, it might be necessary to sample the boundary of the closed volume adaptively to increase the speed of convergence.

\begin{table}[h]
\centering
\begin{tabular}{c|cc}
    $m$ &  QMC+LT & QMC+LT+CS \\
\hline
  $2$ & $0.7804$ &   $1.1830$     \\
  $3$ & $0.6820$ &   $1.0383$    \\
  $4$ & $0.6815$ &   $0.7925$    \\
  $5$ & $0.5894$ &   $0.7794$    \\
  $6$ & $0.5328$ &   $0.6015$    \\
  $\vdots$ & $\vdots$ &   $\vdots$    \\
  $60$ & $0.5454$ &   $0.5040$     
\end{tabular}
\caption{The convergence factors $\alpha$, which we obtained from the linear regression $\log(\sigma)=\beta-\alpha\log(N)$ for the binary barrier.}
\label{table:conv}
\end{table}



\section{Extension: Knock-in options}\label{sec:knock-in}

In the previous section we have constructed the conditional sampling scheme for the LT algorithm for knock-out options, i.e., the value is set to zero when crossing the barrier(s).
In this section, we will extend our algorithm to knock-in options where the option is worthless unless the barrier is crossed.
Such an extension was also considered in \cite{GS2001} where the authors note ``Dealing with knock-in options is not simple, but is possible if there is a known expression $g_i(\bsS_i)$ for the present value of a barrierless option, received at time $t_i$ when the state vector is $\bsS_i$, whose payoff will be $g(\bsS_m)$ at time $t_m$.''
They then continue with ``This is the case for sufficiently simple knock-in options'' and give some examples.
We will first discuss the method described in \cite{GS2001} and then propose an easier and more flexible scheme which can also be used when this condition is not fulfilled.
Note that in \cite{GS2001} paths are constructed in an incremental fashion as this is key to understanding their procedure.

We consider an up-\&-in condition, i.e., a payoff
\begin{align*}
 g(S_1(t_1),\ldots,S_1(t_m), \ldots, S_n(t_m))
 &=
 \max(f(S_1(t_1),\ldots,S_n(t_m)),0)\,\mathbb{I}\left\{\max_{j=1,\ldots,m}S_1(t_j)>B\right\}
 .
\end{align*}
We first explain the scheme from \cite{GS2001}.
Assume that there is a known expression for the option value without the barrier at every time step $t_i$ given the state vector $\bsS_i = (S_1(t_1),\ldots,S_1(t_i),\ldots, S_n(t_i))$ (an incomplete path from $t_0$ up to $t_i$), denoted by $g(\bsS_i)$,
then by noting that a knocked-in barrier option becomes a regular option on knock-in, one samples the variable
\begin{align*}
	\sum_{i=1}^m \left(\mathbb{I}\left\{\max_{j=1,\ldots,i-1}S_1(t_j)<B\right\}-\mathbb{I}\left\{\max_{j=1,\ldots,i}S_1(t_j)<B\right\}\right) g(\bsS_i)
	&=
	\sum_{i=\argmin_{j=1,\ldots,m} S_1(t_j) \ge B}^m
	g(\bsS_i)
	.
\end{align*}
To take advantage of importance sampling the authors in \cite{GS2001} extend this basic sampling scheme so that in each time step $t_i$, there are two successors to $S_1(t_i)$: $S_1(t_{i+1})$ simulated conditional on \emph{no} knock-in and $S^*_1(t_{i+1})$ simulated conditional on knock-in.
(The incremental path is afterwards continued from the $S_1(t_{i+1})$ value.)
The authors then propose to sample
\begin{align*}
	&
	\sum_{i=1}^m L_{i-1} \, \left( 1- \mathbb{P}\left[ S_1(t_{i-1}) < B \mid S_1(t_j) \right] \right) \, g(\bsS^*_i)
	,
\end{align*}
where $L_{i-1}$ is the likelihood as defined in\RefEq{eq:likelihood}.
Now each term in the sum contributes to the random variable.
However, this scheme is rather restrictive as it is required to know the value of the option in each time step.
When considering a knock-in arithmetic Asian option for example, this is not the case.

We propose an easier and more flexible scheme as an extension of our method from \RefSec{sect:LTCS}. Under the LT algorithm, when considering an up-\&-in option, we must have for some $t_j$
\begin{align*}
 \sigma_1W_1(t_j) 
  &\geq
  b(t_j)
  ,
\end{align*}
where $b(t_j)$ is defined in\RefEq{eq:LTrestric}.
The difference with a knock-out type option is that now the above condition must only hold for at least one $t_j$.
Again suppose $a_{j,1}>0$ for $j\in\mathcal{P} \subseteq \{1,\ldots,m\}$ and $a_{j,1}<0$ for all other $j \notin \mathcal{P}$.
For an up-\&-in option we can essentially ``force'' the paths to cross the barrier level by imposing
\begin{align*}
 \Phi^{-1}(u_1)
  &>\min_{j\in\mathcal{P}}\left[
  \frac{b(t_j) - a_{j,2}\Phi^{-1}(u_2) - \ldots - a_{j,mn}\Phi^{-1}(u_{mn})}{a_{j,1}}\right]
\end{align*}
or
\begin{align*}
 \Phi^{-1}(u_1)
  &<\max_{j\notin\mathcal{P}}\left[
  \frac{b(t_j) - a_{j,2}\Phi^{-1}(u_2) - \ldots - a_{j,mn}\Phi^{-1}(u_{mn})}{a_{j,1}}\right]
  .
\end{align*}
For a down-\&-in option the condition on $u_1$ is analogously
\begin{align*}
 \Phi^{-1}(u_1)
  &<\max_{j\in\mathcal{P}}\left[
  \frac{b(t_j) - a_{j,2}\Phi^{-1}(u_2) - \ldots - a_{j,mn}\Phi^{-1}(u_{mn})}{a_{j,1}}\right]
\end{align*}
or
\begin{align*}
 \Phi^{-1}(u_1)
  &>\min_{j\notin\mathcal{P}}\left[
  \frac{b(t_j) - a_{j,2}\Phi^{-1}(u_2) - \ldots - a_{j,mn}\Phi^{-1}(u_{mn})}{a_{j,1}}\right]
  .
\end{align*}
We will denote these bounds again by $\Upsilon_d$ and $\Upsilon_u$.
We can again prove that this estimator is unbiased and has a standard deviation at most that of the unaltered QMC+LT method.
\begin{theorem}
The estimator for knock-in options based on conditional sampling by
\begin{align}
	\hat{g}_2
	&=
	(\Upsilon_u-\Upsilon_d)\max(f(\hat{u}_1,u_2,\ldots,u_{mn}),0)
	\label{Eq:estim}
\end{align}
 is unbiased. That is,
	$\mathbb{E}[g] =\mathbb{E}[\hat{g}_2].$
\end{theorem}
\begin{proof}
Analogous to \RefThm{Thm:LTCSunbiased}.
\end{proof}
\begin{theorem}
	When using a randomly shifted quasi-Monte Carlo method, the estimator defined in\RefEq{Eq:estim} has reduced variance. That is,
	$\var[ \hat{g}_2] \leq \var[g ].$
	The inequality is strict if $\mathbb{P}\left[\max_j S_1(t_j) \le B\right] > 0$ and $\mathbb{E}\left[g\right]>0$, i.e., if there is a positive chance of no knock-in and a positive payoff.
\end{theorem}
\begin{proof}
Analogous to \RefThm{Thm:RedVar}.
\end{proof}

\subsection{Verification by continuous down-\&-in put with knock-in condition}
We can check the valuation of a simple continuous down-\&-in put option to illustrate the unbiasedness of our method for pricing knock-in options.
Consider the put variant:
\begin{align*}
  g(S(t)_{t\in[0,T]})
  &=
  \max\left(K-S(T),0\right) \mathbb{I}\left\{ \min_{t\in[0,T]} S(t) \leq B \right\}
  .
\end{align*}
We approximate the option value by using a fine time discretization of $500$ steps.
The value of the option using a continuous barrier is known analytically, see, e.g., \cite{Wilmott2006}.
The model parameters are $K=100$, $B=80$, $\sigma=20\%$ and $r=5\%$.
Three maturities are considered: $T=1$ year, $T=6$ months and $T=1$ month.
The valuation is based on $4096$ samples and $40$ independent digital shifts. 
The result is shown in \RefFig{fig:DIP}.
We see that the valuations are very close to the analytic values. 

\begin{figure}
	\centering
  \includegraphics[width=0.4\textwidth]{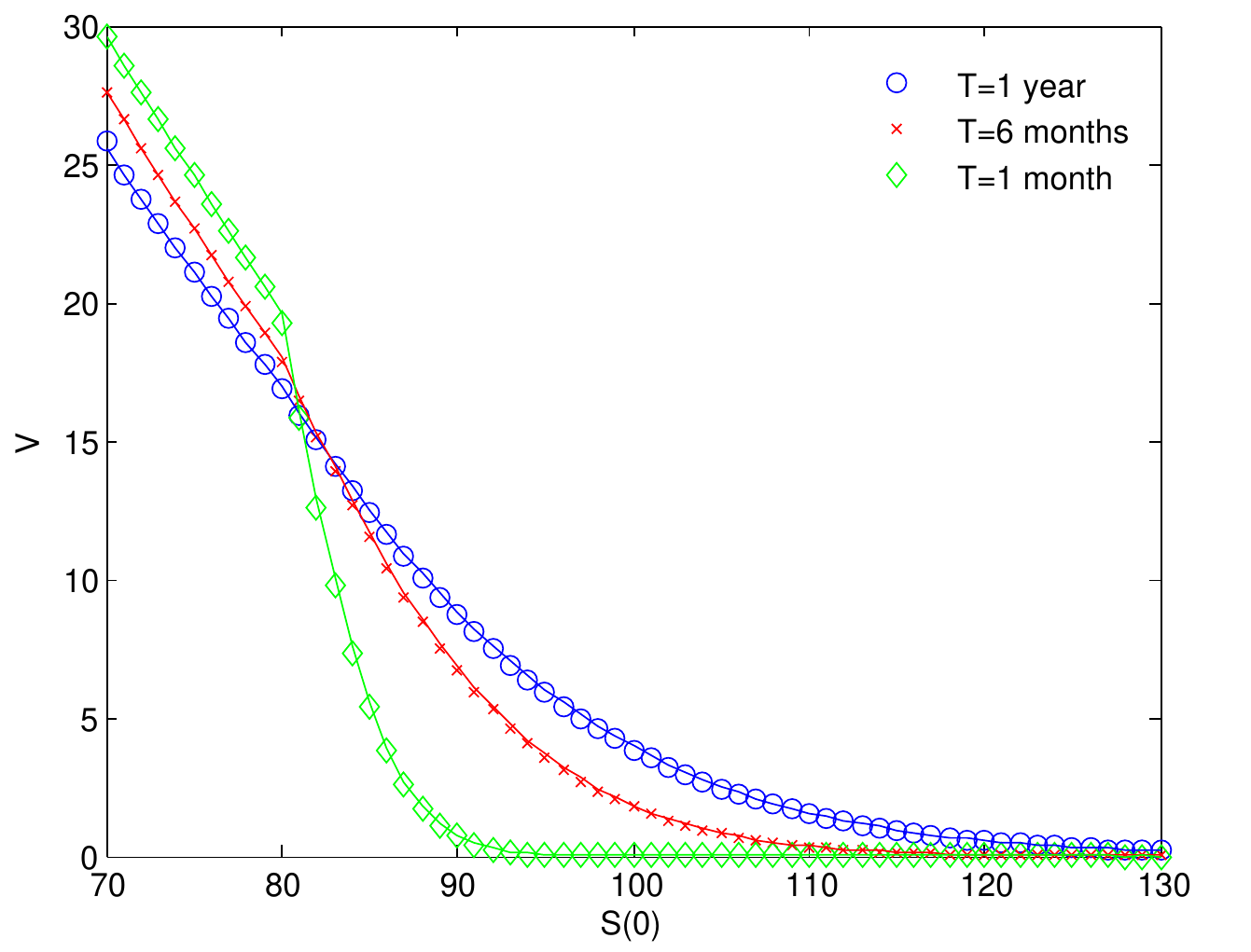}
  \caption{
    Verification with continuous down-\&-in put option with knock-in condition. Shown here are the option values in function of the initial asset price for three different maturities.
  }\label{fig:DIP}
\end{figure}

\subsection{Basket Asian with knock-in condition}

As a second test we consider the same Asian basket on four assets as in \RefSec{ex:sbba} but now with a knock-in condition.
Since there does not exist a closed-form solution of the arithmetic Asian option, we are not able to use the MC+CS method from \cite{GS2001}.
Therefore, we report the ratios of the standard deviations of the value estimates of QMC+LT to QMC+LT+CS for various choices of $\sigma_1$, $K$ and $B$ in \RefTab{table:Ex5}. 
The first three examples take $\sigma_1=25\%$, $K=70$ and $B=105$, $125$ and $200$. When the barrier is close to the starting value of $S_1$ (being $100$), the probability of knocking-in is higher, but we still see a significant variance reduction when using conditional sampling. As the barrier moves further away from the initial asset value, the probability of hitting the barrier drops, and we see that the variance reduction improves even further.

The next two examples take $\sigma_1=55\%$, $B=125$ and $K=70$ or $K=90$. When $K=70$, we see a variance reduction for the QMC+LT+CS method, but we see that increasing $\sigma_1$ from $25\%$ to $55\%$ decreases the variance reduction for $P_1$. This is what we would expect, since a higher volatility for $S_1$ implies that more paths are hitting the barrier without conditional sampling.


Also for the last two examples we observe that QMC+LT+CS shows a variance reduction to the QMC+LT method.
In conclusion: the results of our conditional sampling scheme QMC+LT+CS extended for knock-in options are again very satisfactory. 

\begin{table}
\begin{center}
\begin{tabular}{ cA |cA }
     $(P, \sigma_1, B, K)$ & QMC+LT+CS & $(P, \sigma_1, B, K)$ & QMC+LT+CS   \\
\hline  
$(P_1,0.25,105,70)$ & $287$\% &
$(P_2,0.25,105,70)$ & $237$\% \\
$(P_1,0.25,125,70)$ & $556$\% &
$(P_2,0.25,125,70)$ & $320$\% \\
$(P_1,0.25,200,70)$ & $2272$\% &
$(P_2,0.25,200,70)$ & $345$\% \\[1mm]
$(P_1,0.55,125,70)$ & $297$\%&
$(P_2,0.55,125,70)$ & $383$\% \\
$(P_1,0.55,125,90)$ & $213$\% &
$(P_2,0.55,125,90)$ & $372$\% \\[1mm]
$(P_1,0.25,120,110)$ & $180$\% &
$(P_2,0.25,120,110)$ & $178$\% \\
$(P_1,0.25,110,100)$ & $133$\% &
$(P_2,0.25,110,100)$ & $171$\%                   
\end{tabular}
\caption{Basket Asian with knock-in condition. The reported numbers are the standard deviations of the QMC+LT method (as the MC+CS method does not work for this example) divided by those of the QMC+LT+CS method. Both methods use $4096$ samples and $40$ independent shifts.
	 }
\label{table:Ex5}
\end{center}
\end{table}

\section{Extension: Root-finding}\label{sec:root-finding}

In this section we want to exploit the fact that our construction in \RefSec{sect:LTCS} only modifies the first uniform variate to satisfy the barrier condition. 
Here we propose to exploit the influence this variate has 
by calculating the bounds on $u_1$ as in \RefSec{sect:LTCS} to fulfill the barrier condition and then, using root-finding, determine the region (between these bounds) which produces a positive payoff.
We can then integrate out $u_1$ analytically. 

The idea of using root-finding is quite natural and has been previously discussed in, e.g., \cite{Gerstner2007} and \cite{Holtz2011}, albeit more directly.
For this method to be applicable it is necessary that we are able to determine bounds on $u_1$ (or equivalently $z_1 = \Phi(u_1)$) for it to be applicable.
We will consider an Asian basket as an example of a more complicated payoff and then afterwards a simple put option with a knock-in condition.

\subsection{Single barrier Asian basket}

An Asian basket option on $n$ assets has a payoff
\begin{align*}
  g
  &=
  \max\left(\frac1{mn} \sum_{i_1=1}^n \sum_{i_2=1}^m S_{i_1}(t_{i_2}) - K , 0 \right)
  \mathbb{I}\left\{\cdot\right\},
\end{align*}
with possible some barrier condition for knock-out or knock-in (denoted by $\mathbb{I}\left\{\cdot\right\}$ above).
The following theorem ensures that the problem is well-posed for Asian barrier options.
\begin{theorem}
Given any covariance matrix $\tilde{\Sigma}$ with factorization $AA'$, the function 
\begin{align*}
	f(z_1)
	&= 
	\frac{1}{mn} \sum_{i=1}^{mn} 
	S_{i_1}(0) 
	\, e^{(r-\sigma_{i_1}/2) t_{i_2} + \sum_{k=2}^{mn} a_{ik} z_k}
	\, e^{a_{i1} z_1} 
	-
	K
\end{align*}
\textup{(}where $i_1=\lfloor (i-1)/m\rfloor+1$ and $i_2 = i-(i_1-1)m$\textup{)} has at most two zeros. If all elements $a_{i1}$, $i=1,\ldots,mn$ are of the same sign and $K>0$, then the function $f(z_1)$ has exactly one zero.
\label{Thm:zeros}
\end{theorem}
\begin{proof}
First, consider the case where $a_{i1}\geq 0$ for all $i$.
Then each asset price is a monotonically increasing function of $z_1$.
Therefore, the sum is a monotone function of $z_1$ as well.
Since $S_{i_1}(t_{i_2})\rightarrow 0$ for $z_1\rightarrow -\infty$ and $S_{i_1}(t_{i_2})\rightarrow+\infty$ for $z_1\rightarrow +\infty$, there is exactly only one zero of the function $f(z_1)$ if $K>0$, and none if $K\leq 0$.
A similar argument can be made when $a_{i1}\leq 0$ for all $i$. 

Consider then the case where $a_{i1}$ has mixed signs. 
We have that 
\begin{align}
	f'(z_1)
	&=  
	\frac{d}{dz_1}f(z_1) = \frac{1}{mn}\sum_{i=1}^{mn} S_{i_1}(0) \,
	e^{(r-\sigma_{i_1}/2) t_{i_2} 
	+ \sum_{k=2}^{mn} a_{ik} z_k} 
	\, a_{i1} \, e^{a_{i1} z_1}
	\label{eq:dfz1}
\end{align}
and
\begin{align}
	f''(z_1)
	&= 
	\frac{d^2}{dz_1^2}f(z_1) 
	= 
	\frac{1}{mn}\sum_{i=1}^{mn} S_{i_1}(0) \, 
	e^{(r-\sigma_{i_1}/2) t_{i_2}
	+ \sum_{k=2}^{mn} a_{ik} z_k} 
	\, a^2_{i1} \, e^{a_{i1} z_1}
	.
	\label{eq:ddfz1}
\end{align}
Clearly, $f''(z_1)>0$ everywhere.
The function $f'(z_1)$ can be written as the difference of a monotone increasing function and a monotone decreasing function in $z_1$. 
This implies that $f'(z_1)\rightarrow-\infty$ as $z_1\rightarrow-\infty$ and $f'(z_1)\rightarrow+\infty$ as $z_1\rightarrow+\infty$.
So $f'(z_1)$ changes sign at least once.
For $f(z_1)$ to have more than two zeros, $f'(z_1)$ would have to change sign more than twice. 
However, $f'(z_1)$ has only one zero, which can be easily seen using Rolle's theorem and the observation that $f''(z_1)>0$.
Therefore $f'(z_1)$ changes sign exactly once, which leaves the possibility of zero, one or two zeros.
\end{proof}

In case the elements in the first column of $A$ have the same sign, we are guaranteed there is exactly one zero of the payoff function, which can be found using a root-finding algorithm such as Newton--Raphson. In case the root falls outside the interval $[\Upsilon_d,\Upsilon_u]$ as defined in \RefSec{sect:LTCS}, there is no value for $z_1$ which produces a positive payoff and we set the sample value equal to zero. If, on the other hand, the root does fall inside the interval, we can analytically integrate the payoff function over $z_1$ (or $u_1$) by using the forthcoming lemma.

When the elements in the first column of $A$ have mixed signs, things get a bit more complicated since, given $z_2$ to $z_{mn}$, there might be zero, one or two roots of the payoff function $g$.
We do know however that $f'$ has exactly one zero, so we propose to find this zero using an appropriate root-finding algorithm, and then to check whether $f$ is positive, zero or negative in this point. In case $f$ is positive, there is no zero and we can use the bounds found for the barrier condition as the range of integration. We can do the same if $f$ is zero. If $f$ is negative, we must find two zeros, again using the appropriate root-finding algorithm, initiated with a point to the left and to the right of the zero of $f'$. 

If we call these roots $\Gamma_d$ and $\Gamma_u$, where they can be respectively equal to $-\infty$ and $\infty$, then we can define $\Xi_d = \max(\Upsilon_d, \Gamma_d)$ and $\Xi_u = \min(\Upsilon_u, \Gamma_u)$ an integrate out the first dimension analytically as the next lemma shows.

\begin{lemma}
	Given $z_2$ to $z_{mn}$ and the bounds $\Xi_d(z_2,\ldots,z_{mn})$ and $\Xi_u(z_2,\ldots,z_{mn})$ on $z_1$ for which an Asian barrier option has a positive payoff and satisfies the barrier condition, then the mean value over $z_1$ is given by
	\begin{multline*}
		\mathbb{E}_{z_1}[g(z_1)|z_2,\ldots,z_{mn}]
		\\=
		\frac{1}{mn} \sum_{i=1}^{mn} 
		S_{i_1}(0) 
		e^{(r-\sigma_{i_1}/2) t_{i_2}
		+ \sum_{k=2}^{mn} a_{ik} z_k} 
		e^{a_{i1}^2/2} \left( \Phi(\Xi_u-a_{i1})-\Phi(\Xi_d-a_{i1})\right)
		- (\Phi(\Xi_u)-\Phi(\Xi_d)) K
		.
	\end{multline*}
\label{lem:AInt}
\end{lemma}
\begin{proof}
	We can write 
	\begin{align*}
		\mathbb{E}_{z_1}[g(z_1)|z_2,\ldots,z_{mn}]
		&=
		\sum_{i=1}^{mn} S_{i_1}(0) \, 
		e^{(r-\sigma_{i_1}/2) t_{i_2} + \sum_{k=2}^{mn} a_{ik} z_k} 
		\int_{\Xi_d}^{\Xi_u} \frac{e^{a_{i1} x - x^2/2}}{\sqrt{2\pi}} dx
		-K \int_{\Xi_d}^{\Xi_u} \frac{e^{-x^2/2}}{\sqrt{2\pi}} dx
		.
	\end{align*}
	The result then follows by straightforward calculations.
\end{proof}
The extra cost of using the root-finding procedure is limited: in the worst case (i.e., when the elements in the first column of $A$ have mixed signs) we need to calculate\RefEq{eq:dfz1} and\RefEq{eq:ddfz1}.
We can actually precompute 
\begin{align*}  
    &
    S_{i_1}(0) \,
	e^{(r-\sigma_{i_1}/2) t_{i_2} 
	+ \sum_{k=2}^{mn} a_{ik} z_k}, 
\end{align*}
which has to be calculated in the original algorithm as well, and then use vector multiplication with $(a_{i1}^{\tau} \, e^{a_{i1}} )_i$, where the power of $a_{i1}$ corresponds to the $\tau$th derivative of $f$. 
The original algorithm required one such multiplication, so if we take for example four steps for our root-finder, we need three multiplications on top of the original algorithm.
Higher derivatives can be computed at almost no extra cost, since the only thing that changes is the power $\tau$. 
In our algorithm, we use a fourth-order Newton--Raphson algorithm with 10 iteration steps.

At first thought one might think that having a negative correlation might imply different signs in the first column of $A$.
Surprisingly, as shown in \RefSec{sec:sba2} this relation does not hold necessarily.
Three examples of the root-finding problem are shown in \RefFig{fig:rho_AInt}, where we consider an Asian up-\&-out option with parameters $S_1(0)=S_2(0)=100$, $r=5\%$, $\sigma_1=30\%$, $\sigma_2=40\%$, $B=110$, $K=100$, $T=1$, $m=260$ and different correlations $\rho$ between the two assets to display some possible cases for the roots of $f(z_1)$. \RefFig{fig:rho_neg_30} is an example with negative correlation $\rho=-30\%$ but only one zero.
 
\begin{figure}[ht]
\centering
\subfigure[$\rho=-90\%$]{
\includegraphics[scale=0.3]{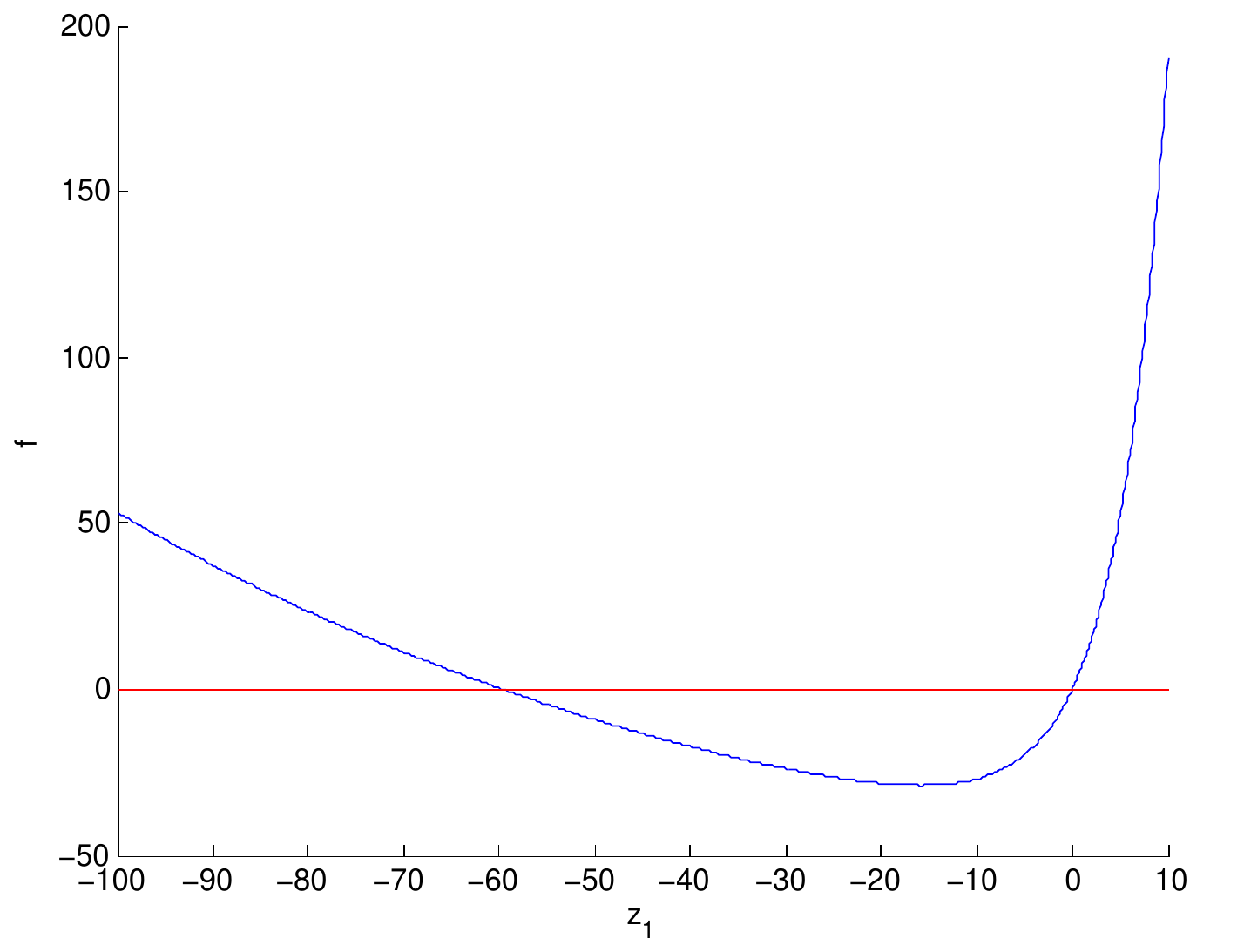}
\label{fig:rho_neg_90}
}
\subfigure[$\rho=-30\%$]{
\includegraphics[scale=0.3]{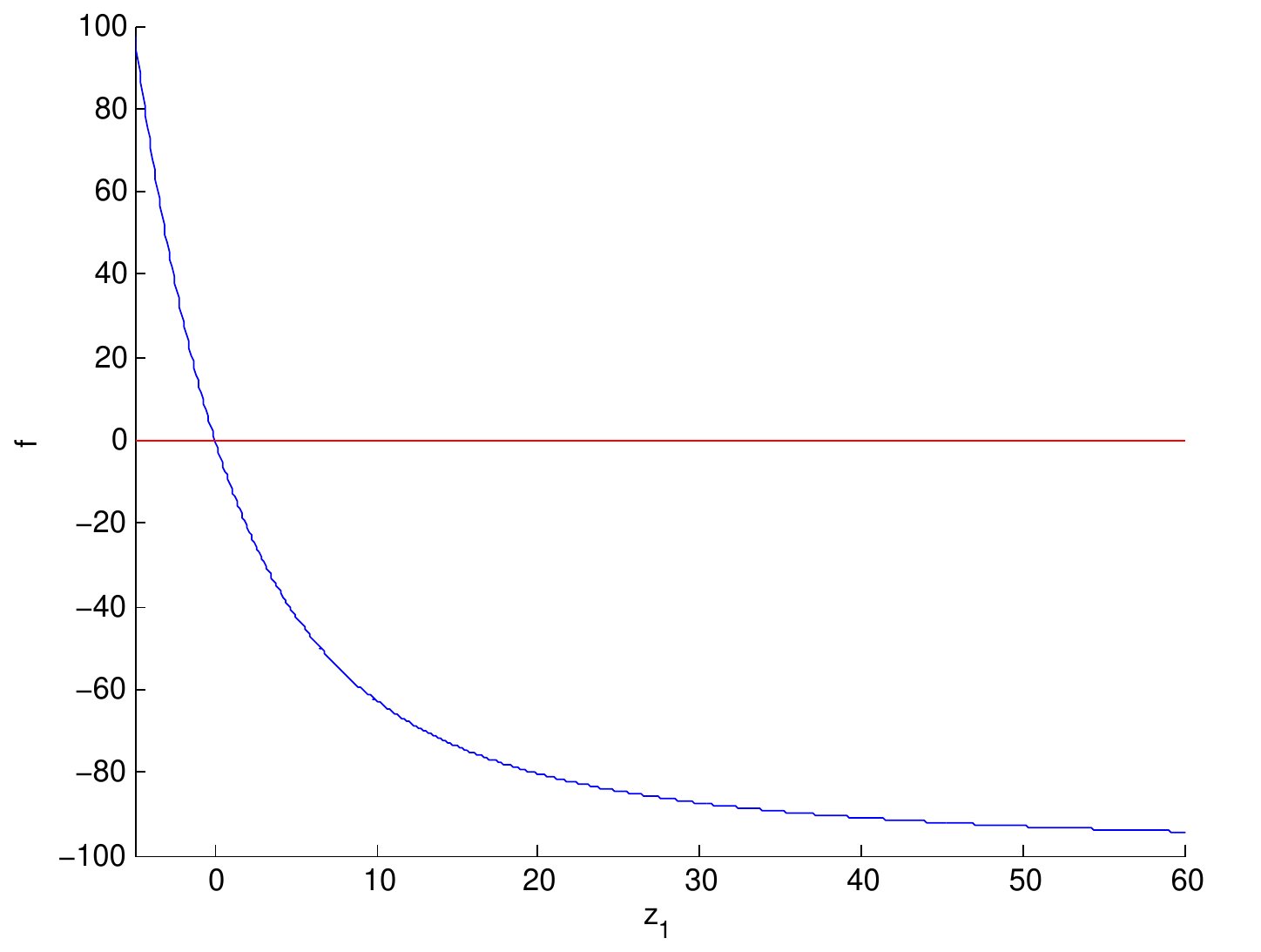}
\label{fig:rho_neg_30}
}
\subfigure[$\rho=80\%$]{
\includegraphics[scale=0.3]{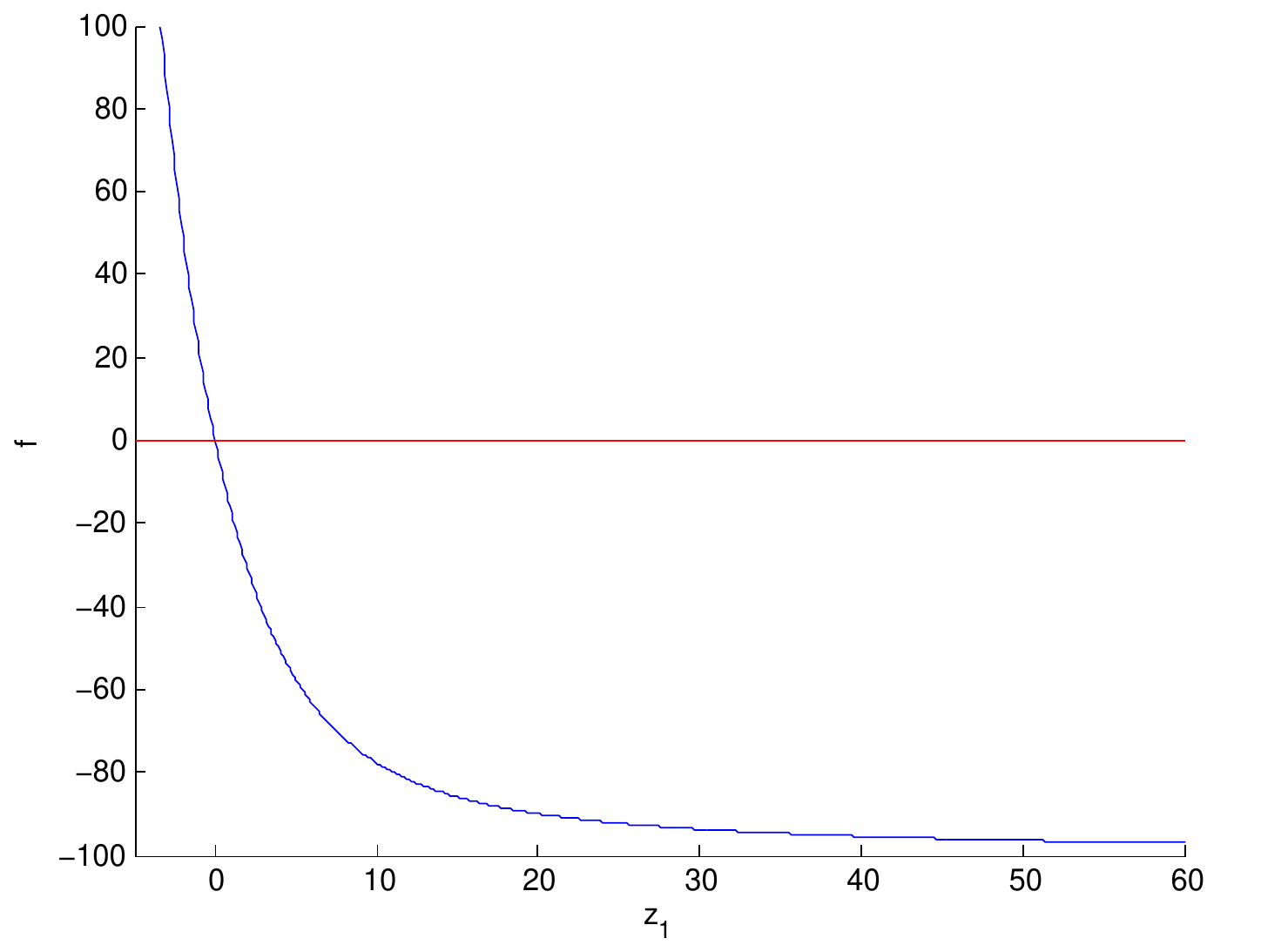}
\label{fig:rho_pos_80}
}
\caption{Possible root configurations for an Asian basket on two stocks with different correlation.}\label{fig:rho_AInt}
\end{figure}

For our numerical test we look again at the single barrier Asian basket from \RefSec{ex:sbba}.
The results for different choices of $K$, $B$ and $\sigma_1$ are shown in \RefTab{table:Ex6}. 
By QMC+LT+CS+RF we denote the new algorithm from this section.
There is a significant variance reduction compared to the other methods for all parameter choices.

\begin{table}
\begin{center}\footnotesize
\hspace{-5mm}\begin{tabular}{ c B A | c B A}
     $(P, \sigma_1, B, K)$ & QMC+LT+CS+RF & QMC+LT+CS &
     $(P, \sigma_1, B, K)$ & QMC+LT+CS+RF & QMC+LT+CS 
\\
\hline  
$(P_1,0.25,125,100)$ & $2039$\% & $958$\% &
$(P_2,0.25,125,100)$ & $1955$\% & $1172$\% \\
$(P_1,0.25,110,100)$ & $960$\% & $446$\% &
$(P_2,0.25,110,100)$ & $1172$\% & $595$\% \\
$(P_1,0.25,105,100)$ & $638$\% & $263$\% &
$(P_2,0.25,105,100)$ & $761$\% & $367$\% \\
$(P_1,0.25,110,90)$ & $910$\% & $737$\% &
$(P_2,0.25,110,90)$ & $908$\% & $782$\% \\
$(P_1,0.25,105,90)$ & $757$\% & $576$\% &
$(P_2,0.25,105,90)$ & $623$\% & $581$\% \\
$(P_1,0.25,125,110)$ & $1923$\% & $489$\% &
$(P_2,0.25,125,110)$ & $4693$\% & $683$\% \\
$(P_1,0.55,125,100)$ & $939$\% & $437$\% &
$(P_2,0.55,125,100)$ & $1082$\% & $535$\% \\
$(P_1,0.55,125,110)$ & $1035$\% & $234$\% &
$(P_2,0.55,125,110)$ & $1381$\% & $310$\% \\                         
\end{tabular}
\caption{Single barrier Asian basket. The reported numbers are the standard deviations of the MC+CS method divided by those of the 
QMC+LT+CS and QMC+LT+CS+RF methods. The MC+CS method uses $163840$ samples, while the QMC+LT+CS and QMC+LT+CS+RF methods use $4096$ samples and $40$ independent shifts.
}
\label{table:Ex6}
\end{center}
\end{table}

\subsection{Down-\&-in put option}

As a final example we consider a simple put option with a knock-in condition:
\begin{align*}
  g
  &=
  \max\left(K-S(t_m), 0\right)
  \mathbb{I}\left\{ \min_{t_i=t_1,\ldots,t_m} S(t_i)\leq B \right\}
  .
\end{align*}
We fix $m=130$, $S(0)=100$, $r=5\%$ and $T=6$m.
\RefTab{table:ExPut} shows the results for different values of $\sigma$, $B$ and $K$.
In this example of a knock-in option, we are able to force $100\%$ of the samples to result in positive payoffs. 
For a wide variety of parameter choices, including different volatilities and both in-the-money and out-the-money barriers and options we see significant gains compared to the QMC+LT+CS and QMC+LT methods.
\begin{table}
\begin{center}
\begin{tabular}{ c B A}
     $(\sigma, B, K)$ & QMC+LT+CS+RF & QMC+LT+CS  \\
\hline 
$(0.25,90,100)$ & $442$\% & $184$\%\\
$(0.25,90,90)$ & $551$\% & $134$\%\\
$(0.25,90,80)$ & $742$\% & $176$\%\\
$(0.25,75,80)$ & $884$\% & $480$\%\\
$(0.55,75,80)$ & $575$\% & $169$\%\\
$(0.55,70,60)$ & $657$\% & $185$\%
\end{tabular}
\caption{Down-\&-in put option. The reported numbers are the standard deviations of the QMC+LT method divided by those of the QMC+LT+CS and QMC+LT+CS+RF methods. All methods use $4096$ samples and $40$ independent shifts.
	 }
\label{table:ExPut}
\end{center}
\end{table}


\section{Conclusion and outlook}\label{sec:conclusion}

In this paper we have devised a conditional sampling method for QMC sampling under the LT algorithm for barrier options.
Our method tries to satisfy the barrier condition by modifying the first uniform variable used to construct the sample paths.
We have shown that the method is unbiased for randomized QMC methods and has a variance which is lower than or equal to the method without using conditional sampling.
Furthermore, the method can be used for both knock-out and knock-in conditions with the same basic principle of modifying the first uniform variable.
When adding a root-finding method we can additionally also satisfy to have always a positive payoff which further improves the performance of the method.
Extensive numerical results show the effectiveness of the new method.

While we focussed on the Black--Scholes framework in this paper, it is natural to expand our method to other market models as well.
For the Heston model we refer to the follow-up paper \cite{NicoHESLT}.

\end{document}